\documentclass[10pt, final, journal, letterpaper, twocolumn]{IEEEtran}
\makeatletter
\def\ps@headings{%
\def\@oddhead{\mbox{}\scriptsize\rightmark \hfil \thepage}%
\def\@evenhead{\scriptsize\thepage \hfil \leftmark\mbox{}}%
\def\@oddfoot{}%
\def\@evenfoot{}}
\makeatother 

\IEEEoverridecommandlockouts
\usepackage{bbm}
\usepackage{amsfonts}
\usepackage[dvips]{graphicx}
\usepackage{times}
\usepackage{cite}
\usepackage{amsmath}
\usepackage{array}
\usepackage{amssymb}

\usepackage{stfloats}
\usepackage{slashbox}
\usepackage{graphicx}
\usepackage{footnote}
\usepackage{amsthm}
\usepackage{booktabs}
\usepackage{array}
\usepackage{algorithmic}
\usepackage{algorithm}
\usepackage{subeqnarray}
\usepackage{cases}
\usepackage{threeparttable}
\usepackage{color}
\usepackage{epstopdf}
\usepackage{bbm}
\usepackage{bm}
\usepackage{multirow}
\usepackage[labelformat=simple]{subcaption}
\usepackage{adjustbox}

\newtheorem{lemma}{Lemma}

\allowdisplaybreaks

\begin{document}
\title{Collaborative  Multi-Agent Multi-Armed Bandit Learning for Small-Cell Caching}
\author{\IEEEauthorblockN{Xianzhe Xu, Meixia Tao, \IEEEmembership{Fellow,~IEEE}, and Cong Shen, \IEEEmembership{Senior Member,~IEEE}}\\
\thanks{X. Xu and M. Tao are with the Department of Electronic Engineering, Shanghai
Jiao Tong University, Shanghai 200240, China (Emails: august.xxz@sjtu.edu.cn;
mxtao@sjtu.edu.cn). C. Shen is with the Charles L. Brown Department of Electrical and Computer Engineering, University of Virginia, Charlottesville, VA 22904, USA. (Email: cong@virginia.edu). Part of this work was presented in IEEE Globecom 2018 \cite{xxz_123}. This work is supported in part by the NSF of China under grant 61941106 and by the National Key R$\&$D Project of China under grant 2019YFB1802702.}}

\maketitle
\vspace{-1.5cm}
\begin{abstract}
This paper investigates learning-based caching in small-cell networks (SCNs) when user preference is unknown. The goal is to optimize the cache placement in each small base station (SBS) for minimizing the system long-term transmission delay. We model this sequential multi-agent decision making problem in a multi-agent multi-armed bandit (MAMAB) perspective. Rather than estimating user preference first and then optimizing the cache strategy, we propose several MAMAB-based algorithms to directly learn the cache strategy online in both stationary and non-stationary environment. In the stationary environment, we first propose two high-complexity agent-based collaborative MAMAB algorithms with performance guarantee. Then we propose a low-complexity distributed MAMAB which ignores the SBS coordination. To achieve a better balance between SBS coordination gain and computational complexity, we develop an edge-based collaborative MAMAB with the coordination graph edge-based reward assignment method. In the non-stationary environment, we  modify the MAMAB-based algorithms proposed in the stationary environment by proposing a practical initialization method and designing new perturbed terms to adapt to the dynamic environment. Simulation results are provided to validate the effectiveness of our proposed algorithms. The effects of different parameters on caching performance are also discussed.
\end{abstract}
\section{Introduction}
%
%
%
In recent decades, the proliferation of mobile devices results in a steep growth of mobile data traffic, which imposes heavy pressure on backhaul links with limited capacity in cellular networks. Due to the fact that only a small number of files accounts for the majority of the data traffic, caching popular files at small base stations (SBSs) has been widely adopted to reduce the traffic congestion and alleviate the backhaul pressure \cite{wang2014cache,bastug2014living}. Many previous works have studied the cache placement problem with various cache strategies and performance metrics. In \cite{peng2015backhaul,shanmugam2013femtocaching,bacstu?2015cache}, the authors focus on the deterministic caching that SBSs cache files entirely without partitioning. The works \cite{blaszczyszyn2015optimal,chae2016caching,7562510} consider the probabilistic caching that SBSs cache files according to certain well-designed probabilities. The authors in \cite{6763007,xu2017modeling,8017548} consider the coded caching that files are partitioned into multiple segments, and these segments after coding are cached at SBSs. However, these works on wireless caching rely on a strong assumption that file popularity or user preference is perfectly known in advance, which is somehow unrealistic in general. Therefore, caching design without the knowledge of file popularity or user preference is a necessary but challenging issue. The aim of this work is to address the above issue by investigating the caching design when user preference is unknown.

Several works have studied the wireless caching problem without the knowledge of file popularity or user preference. Some of them tackle this problem by first estimating file popularity (or the number of requests) and then optimizing the cache strategy based on the estimated file popularity (or estimated number of requests) \cite{7524380,li2016trend,blasco2014learning,blasco2014multi,Sengupta2014,song2017learning,8466606,8241758}. The authors in \cite{blasco2014learning,blasco2014multi} estimate the file popularity by using multi-armed bandit (MAB) and then the SBS caches the most popular files in the single SBS scenario. In \cite{Sengupta2014}, the authors utilize MAB to estimate the local popularity and then optimize the coded caching in the multiple SBSs scenario. The work \cite{song2017learning} also utilizes MAB to estimate local popularity of each SBS and then optimize the cache strategies to minimize the total cost of multiple SBSs with centralized and distributed algorithms. In \cite{8241758}, the authors estimate the local popularity and global popularity first, and then utilize Q-learning to choose cache actions with local popularity, global popularity and cache strategy as state.

{Note that all these works \cite{blasco2014learning,blasco2014multi,Sengupta2014,7524380,song2017learning,li2016trend,8466606,8241758} only focus on the file popularity (local or global) or the requests collected from multiple users in a cell since they either assume that different users have the same preference or adopt cache hit as the performance metric. However, different users {may} have different preferences in practice. Moreover, {adjacent} SBSs need to design {their} cache strategies collaboratively {as they can communicate with a common set of users}. If we consider file popularity from a set of users rather than user preference, SBSs cannot know whether these requests are from the users that are only covered by one SBS or the users that are covered by multiple SBSs. Therefore, SBSs cannot design the cache strategy collaboratively. Thus, some recent works}\cite{muller2017context,jiang2018user,8510864,8425746} focus on the caching strategy design based on user preference rather than file popularity when user preference is unknown. Recently, the work \cite{muller2017context} utilizes contextual MAB to learn the context-specific content popularity online and takes the dependence of user preference on their context into account. With the estimated context-specific content popularity, SBSs cache files to maximize the accumulated cache hits. In \cite{jiang2018user}, the authors propose an online content popularity prediction algorithm by leveraging the content features and user preference. Then they propose two learning-based edge caching architectures to design the caching policy. Though the authors in \cite{muller2017context,jiang2018user} consider the multiple SBSs scenario, the correlations among SBSs are not considered. In \cite{8510864}, the authors propose a new Bayesian learning method to predict personal preference and estimate the individual content request probability. Based on the estimated individual content request probability, a reinforcement learning algorithm is proposed to optimize the deterministic caching by maximizing the system throughput. The work \cite{8425746} first utilizes the expectation maximization method to estimate the user preference and activity level based on the probabilistic latent semantic analysis model. A low-complexity algorithm is then developed to optimize the D2D cache strategies by maximizing the offloading probability based on the estimated user preference and activity level. However, {a training phase} is required in \cite{8510864,8425746}, which cannot well adapt to the time-varying environment, comparing to online methods. Due to the fact that context information of users is private and the historical {request} number of a single user is often very limited, online estimation of user preference is very difficult in practice. Therefore, some works try to make caching decisions directly, rather than estimating user preference first and then optimizing cache, to minimize the long-term cost or maximize the long-term reward by using reinforcement learning \cite{8491205,guo2017caching} and deep reinforcement leaning \cite{8362276}. However, these works \cite{8491205,guo2017caching,8362276} only focus on the single-cell performance and no coordination among multiple SBSs exists. To our best knowledge, an online collaborative caching problem without the prior knowledge of user preference at the user level has not been well studied in the literature.

In this paper, we aim to optimize the collaborative caching among multiple SBSs by minimizing the {accumulated transmission delay over a finite time horizon} in cache-enabled wireless networks when user {preference is} unknown. The cache strategies need to be optimized collaboratively since each user can be covered by multiple SBSs and experiences different delays when it is served by different SBSs. {Thus, SBSs need to learn and decide the cache strategies collaboratively at each time slot over a finite time horizon. Since earlier decision of caching would shape our knowledge of caching rewards that our future action depends upon, this problem is fundamentally a sequential multi-agent decision making problem. Rather than estimating user preference first and then optimizing the cache strategy,
we learn the cache strategies directly at SBSs online by utilizing multi-agent MAB (MAMAB). This is because estimating the preference of each individual user is difficult due to the limited number of requests from a single user. Besides, contextual information of users, such as age and gender may not be utilized to estimate user preference due to privacy concerns. On the other hand,  the caching reward in the MAMAB framework considered in our work is defined as the transmission delay reduction compared to the case without caching. It represents the aggregate effect of all user requests and hence can be learned more accurately and with faster speed.
Our prior conference paper \cite{xxz_123} considered the distributed and edge-based collaborative MAMAB algorithms only in the stationary environment, without the theoretical proof of performance guarantee.} The main contributions and results of this {journal version} are summarized as follows.

\begin{itemize}
	\item{We formulate the collaborative caching optimization problem to minimize the {accumulated} transmission delay {over a finite time horizon} in cache-enabled wireless networks without the knowledge of user preference. By defining the reward of caching, we first model this sequential decision making problem in a MAMAB perspective, which is a classical reinforcement learning framework. We then solve this MAMAB problem in both stationary and non-stationary environment.}
	\item{In the stationary environment, we first propose two high-complexity agent-based collaborative MAMAB algorithms with the performance guarantee that the regret is bounded by $\mathcal{O}(\log{(T_{\text{total}})})$, where $T_{\text{total}}$ is the {total number of time slots.} To reduce the computational complexity, we also propose a distributed MAMAB algorithm, which ignores the SBSs coordination. By taking both SBS coordination and computational complexity into account, we propose an edge-based collaborative MAMAB algorithm based on the coordination graph edge-based reward assignment method. We also modify the above MAMAB algorithms by designing new perturbed terms to achieve a better performance.}
\item{In the non-stationary environment, we modify the MAMAB algorithms {proposed in the stationary environment by proposing a practical initialization method and designing new perturbed terms in order} to adapt to the dynamic environment.}

\item{Simulation results are provided to demonstrate the effectiveness of our proposed algorithms in both stationary (synthesized data) and non-stationary environment (Movielens data set) by comparing with some benchmark algorithms. The effects of communication distance, cache size {and user mobility on caching performance} are also discussed.}
\end{itemize}


%
%
%

The rest of the paper is organized as follows. The system model is introduced in Section \uppercase\expandafter{\romannumeral2}. Then we formulate the problem in Section \uppercase\expandafter{\romannumeral3}. In Section \uppercase\expandafter{\romannumeral4} and Section \uppercase\expandafter{\romannumeral5}, we solve this problem in stationary and non-stationary environment, respectively. We then present the simulation results in Section \uppercase\expandafter{\romannumeral6}. Finally, we conclude the paper in Section \uppercase\expandafter{\romannumeral7}.

\emph{Notations}: This paper uses bold-face lower-case $\mathbf{h}$ for vectors and bold-face uppercase $\mathbf{H}$ for matrices. $\mathbf{0}_{m\times 1}$ denotes the $m\times 1$ zero vector. $\mathbb{I}\{X\}$ denotes the indicator operator that $\mathbb{I}\{X\}=1$ if the event $X$ is true and $\mathbb{I}\{X\}=0$ otherwise.
\section{System Model}
\subsection{Network Model}
We consider a SCN where cache-enabled SBSs and users are {located} in a finite region as illustrated in Fig. \ref{fig:coordination_graph}. The sets of users and SBSs are denoted as $\mathcal{U}=\{1,2,\ldots,U\}$ and $\mathcal{M}=\{1,2,\ldots,M\}$, respectively. We assume that SBSs have a limited communication distance $l_c$, which means that SBS $m$ can communicate with user $u$ if the distance between them $l_{m,u}$ does not exceed $l_c$. Thus, we define the neighbor SBSs of user $u$ as the set of SBSs that can communicate with user $u$, which is denoted as $\mathcal{M}_u=\{m\in\mathcal{M}|l_{m,u}\leq l_c\}$. We sort the distance $l_{m,u}$ for $m\in\mathcal{M}_u$ in an increasing order such that {$j_u$} denotes the index of the $j$-th nearest SBS to user $u$. Similarly, the set of the neighbor users of SBS $m$ is denoted as $\mathcal{U}_m=\{u\in\mathcal{U}|l_{m,u}\leq l_c\}$.
\begin{figure}[htb]
\begin{centering}
\includegraphics[scale=.45]{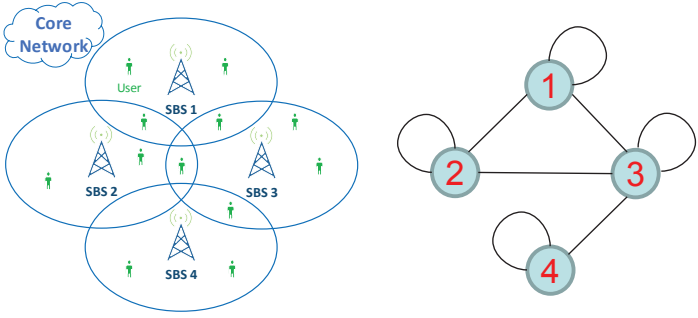}
\caption{\small{System model and its corresponding coordination graph.}} \label{fig:coordination_graph}
\end{centering}
\vspace{-0.4cm}
\end{figure}
\subsection{Cache Model}
There is a file library $\mathcal{F}$ with size $|\mathcal{F}|$ and all files are assumed to be the same normalized size of $1$\footnote{{For the unequal size case, we can divide each files into chunks of equal size, so the analysis and algorithms in this paper still can be applied.}}. Each SBS is equipped with a local cache with size $S<|\mathcal{F}|$. Define a cache matrix $\mathbf{A}^t\in\{0,1\}^{M\times |\mathcal{F}|}$ to denote the cache strategy of SBSs at time slot $t$. The cache matrix should satisfy the cache size constraint $\sum_{f=1}^{|\mathcal{F}|}a_{m,f}^t\leq S$ for $m\in\mathcal{M}$ at each time slot.
\subsection{Service Model}
We adopt the total transmission delay as the performance metric. When SBS $m$ serves its neighbor user $u$ by transmitting a file with unit size, the user will experience a transmission delay $d_{m,u}$. {We model the transmission delay in the noise-limited network and hence interference can be neglected. By considering the large-scale fading only,} the signal-to-noise ratio (SNR) between user $u$ and SBS $m$ is given by $\text{SNR}_{m,u}=\frac{Pl_{m,u}^{-\alpha}}{\sigma_N^2}$,
where $P$ is the transmission power of SBSs, $\alpha>2$ is the path loss exponent and $\sigma_N^2$ is the Gaussian white noise power. The transmission delay between user $u$ and its neighbor SBS $m\in \mathcal{M}_u$ by transmitting a file with unit size is $d_{m,u}=\frac{1}{W\log_2(1+\text{SNR}_{m,u})}$,
where $W$ is the bandwidth. Note that users will experience a larger transmission delay if they are served by farther SBSs. For the non-neighbor SBSs $m\notin \mathcal{M}_u$, the transmission delay $d_{m,u}=\infty$ {since the communication link cannot be established between them.}

Each user requests for files according to its own preference. {The requests of user $u$ at time slot $t$ is denoted as a set $\mathcal{Q}_u^t$ with $|\mathcal{Q}_u^t|\in\mathbb{N}$. This means that a user $u$ can request any number of files at a time slot since a time slot contains several hours or even one day in practice. Therefore, SBSs can satisfy multiple requests at a time slot.} For a user $u$ requesting file $f$ at a given time instance, it will be served by the nearest neighbor SBS that caches file $f$. If none of its neighbor SBSs caches file $f$, the core network will satisfy this request, which induces a much larger transmission delay $d_0$ comparing to the delay when requests are served by SBSs locally, i.e., $d_0\gg d_{m,u}$ for $m\in\mathcal{M}_u$. {Note that when $l_c$ is large enough, all users can be connected to all SBSs in $\mathcal{M}$. This fully connected network is actually a special case of our problem. All parameters are summarized in Table \ref{tab2}.}

\begin{table}[h]
\vspace{-0.2cm}
\centering
\caption{{Notation Table}} \label{tab2}
\begin{tabular}{|c|c|}
\hline
\textbf{Notation} & \textbf{Definition} \\
\hline
$\mathcal{M}$ & SBSs set    \\
$\mathcal{U}$ & Users set   \\
$\mathcal{F}$ & Files set   \\
$\mathcal{M}_u$  & Neighbor SBSs set of user $u$   \\
$\mathcal{U}_m$  & Neighbor users set of SBS $m$   \\
$\mathbf{A}$  & Cache matrix   \\
$\mathcal{Q}_u$  & Requests set of user $u$  \\
$S$  & Cache size   \\
$j_u$            & Index of $j$-th nearest SBS of user $u$ \\
$l_c$            & Communication distance \\
$d_{m,u}$            & Transmission delay between SBS $m$ and user $u$ \\
$d_0$            & Transmission delay of the core network \\
\hline
\end{tabular}
\vspace{-0.4cm}
\end{table}

\section{Problem Formulation}
Our goal is to design cache strategies $\mathbf{A}^t$ online to minimize the accumulated transmission delay over a time horizon $T_{\text{total}}$ when {user preference is unknown.} In this section, we model this cache optimization problem as a sequential multi-agent decision making problem and reformulate this problem in a MAMAB perspective.

{For a user $u$ that requests file $f$ at time slot $t$, it can be satisfied by either one neighbor SBS or the core network. If at least one neighbor SBS that caches the requested file of the user, the user will be served by the nearest SBS that caches the requested file. In this case, the transmission delay depends on the nearest neighbor SBS that caches the requested file, regardless of other farther SBSs. If none of the neighbor SBSs caches the requested file, i.e., $a_{m,f}^t=0$ for all $m\in\mathcal{M}_u$, the user will be served by the core network and experience a much larger delay $d_0$. Since we focus on the transmission delay from the user perspective, the total transmission delay of the network at time slot $t$ is the sum of all individual delays for all users, which is given by:}
\begin{align}
D_t&=\sum_{u=1}^U \sum_{f=1}^{|\mathcal{F}|}\mathbb{I}\{f\in\mathcal{Q}_u^t\}\bigg[\sum_{j=1}^{|\mathcal{M}_u|}a_{j_u,f}^td_{j_u,u}\prod_{n=1}^{j-1}(1-a_{n_u,f}^t) \nonumber\\
&~~~~~~~+\underset{m\in\mathcal{M}_u}{\prod}(1-a_{m,f}^t)d_0\bigg]
.  \label{eqn:delay_general}
\end{align}

We aim to minimize the accumulated transmission delay of the network over the time horizon $T_{\text{total}}$ by {learning and optimizing} the cache strategy at each time slot $t$ when the requests $\mathcal{Q}_u^t$ of all users are not known in advance. The optimization problem is formulated as:
\addtocounter{equation}{1}
\begin{align}
\text{\textbf{P1:}} \quad \underset{\{\mathbf{A^t}\}}{\text{min}} \quad & \sum_{t=1}^{T_{\text{total}}} D_t  \tag{\theequation a} \label{eqn:accumu_delay}   \\
\text{s.t} \quad &\sum_{f=1}^{|\mathcal{F}|} a_{m,f}^t \leq S,~\forall m\in\mathcal{M}~\text{and}~ t=1,2,\ldots,T_{\text{total}}, \tag{\theequation b} \label{eqn:cons11}   \\
 &\mathbf{A}^t\in\{0,1\}^{M\times |\mathcal{F}|}, ~~t=1,2,\ldots,T_{\text{total}}.     \tag{\theequation c}  \label{eqn:cons22}
\end{align}

{Note that the optimization problem $\textbf{P1}$ is a classical caching optimization problem and has been solved if user requests $\mathcal{Q}_u$ and transmission delays between SBSs and users $d_{m,u}$ are perfectly known in advance, which is very difficult in practice. Most traditional learning-based caching works estimate the file popularity first and then optimize the cache strategy since they adopt cache hit rate as the performance metric and assume that different users have the same preference. While for optimization problem $\textbf{P1}$, estimating requests of a single user online is much more difficult than estimating file popularity since the request number of a single user is much smaller than the request number from multiple users in a cell. Besides, users contextual information, such as age and gender, are private and we cannot utilize them to group users and learn preference from a group of users as work \cite{muller2017context} does.\footnote{To our best knowledge, there is no work that estimates user preference $p_{u,f}$ online. References \cite{8510864,8425746} obtain the accurate values of estimated user preference in the training phase, which cannot adapt to time-varying environment well.} Therefore, the conventional predict-then-optimize approach cannot be utilized in this problem. Besides, the optimization problem $\textbf{P1}$ is still NP-complete \cite{shanmugam2013femtocaching} even when the user requests and transmission delays between SBSs and users are perfectly known. Note that the proposed greedy algorithm in \cite{shanmugam2013femtocaching} to solve $\textbf{P1}$ has a high computational complexity since the greedy algorithm needs to exhaustively search all feasible file-SBS pairs and choose the best one at each step. This step is repeated until the caches of all SBSs are full.} Therefore, instead of estimating the user requests first and then optimizing the cache strategy, we would like to learn the cache strategy directly with a low-complexity algorithm to make {sequential caching decisions.}  {Note that reinforcement learning is an effective tool to solve sequential decision making problems. Rather than applying other reinforcement learning algorithms, such as Q-learning, we model this multi-agent sequential decision making problem in a MAMAB perspective since MAB represents a simple class of online learning models with fewer assumptions of the environment. Furthermore, we shall show that the performance of the proposed MAMAB algorithms can be theoretically guaranteed later in Lemma 2.}

Before reformulating \textbf{P1} as a MAMAB problem, we define the reward of caching as the transmission delay reduction compared to the case without caching. {Note that SBSs can observe the transmission delay if they serve their neighbor users with caching files. Besides, all SBSs know the transmission delay $d_0$ when users are served by the core network. Therefore, SBSs can observe the transmission delay and calculate the reward of caching files if they serve their neighbor users.} The reward of SBS $m$ by caching file $f$ at time slot $t$ is given by:
\begin{align}
r_{m,f}^t&=a_{m,f}^t\sum_{u\in\mathcal{U}_m}\mathbb{I}\{f\in\mathcal{Q}_u^t\} \left(d_0-d_{m,u}\right)  \nonumber\\
&~~~~~~~~~~\sum_{j=1}^{|\mathcal{M}_u|}\mathbb{I}\{m=j_u\}\prod_{n=1}^{j-1}(1-a_{n_u,f}^t).  \label{eqn:reward22}
\end{align}
By defining the reward, we can reformulate \textbf{P1} without loss of optimality as:
\begin{align}
\text{\textbf{P2:}} \quad \underset{\{\mathbf{A}^t\}}{\text{max}} \quad & r_{\text{total}}=\sum_{t=1}^{T_{\text{total}}}\sum_{m=1}^M\sum_{f=1}^{|\mathcal{F}|}r_{m,f}^t  \label{eqn:reward123}  \\
\text{s.t} \quad & (\ref{eqn:cons11}), (\ref{eqn:cons22}).   \nonumber
\end{align}

{Optimization problem $\textbf{P2}$ is a MAMAB problem, where each SBS is seen as an agent and each file is seen as an arm. SBSs decide the cache strategy based on its cumulative knowledge to maximize the accumulated reward with unknown distributions of the reward for caching different files. We can solve the optimization problem $\textbf{P2}$} by estimating the reward $r_{m,f}^t$ and then designing cache strategy at each time slot. Actually, $\textbf{P2}$ can be divided into $T_{\text{total}}$ independent sub-optimization problems, one for each time slot, called the single-period optimization (SPO) problem. For each SPO problem, the cache strategy is optimized with the estimated reward $r_{m,f}^t$ based on the historical reward observations. The objective of $\textbf{P2}$ is to maximize the accumulated reward {over a time horizon based on its cumulative knowledge}. Therefore, there exists a tradeoff between exploration (i.e., caching all files a sufficient number of times to estimate the reward more accurately) and exploitation (i.e., caching files to maximize the estimated reward).

For this MAMAB problem, its main difficulty compared to the single agent MAB problem is that there {exists} collisions among different agents (SBSs), i.e., for a user requesting a file, only the nearest neighbor SBS that caches the requested file can serve the user and obtain the reward while {farther} neighbor SBSs obtain no reward even if they cache the requested file. Therefore, {how to make} the caching decisions coordinately is challenging. In this paper, we utilize a coordination graph $G=(V,E)$ \cite{guestrin2002multiagent,kok2006collaborative} to represent the dependencies among SBSs, as shown in Fig. \ref{fig:coordination_graph}. In the coordination graph, each SBS $m\in V$ represents a vertex. An edge $(m,n)\in E$ indicates that the corresponding SBS $m$ and $n$ cover the common users, i.e., $\mathcal{U}_m\cap \mathcal{U}_n\neq\emptyset$, which means that the cache strategies of SBS $m$ and $n$ need to be designed coordinately. Specially, each SBS $m$ has a self edge $(m,m)$. We define the neighbor of the vertex $m$ in the coordination graph as $\Gamma(m)\triangleq \{n\in \mathcal{M}\mid\mathcal{U}_m\cap \mathcal{U}_n\neq\emptyset,m\neq n\}$. As we can see in (\ref{eqn:reward22}), the reward $r_{m,f}^t$ of SBS $m$ by caching file $f$ depends on the {joint cache action} of SBSs $m$ and $\Gamma(m)$, which is denoted as $\textbf{b}_{m,f}^t=[a_{m,f}^t,a_{m_1,f}^t,\ldots,a_{m_{\Gamma(m)},f}^t]$. Therefore, $r_{m,f}^t$ can be rewritten as $r_{m,f}^t(\textbf{b}_{m,f})$.
Due to the fact that the user requests are unknown before designing the cache strategy, we can not analyse the property of the reward $r_{m,f}^t(\textbf{b}_{m,f})$ to solve the optimization problem. In this case, we can utilize MAMAB to estimate the reward $r_{m,f}^t(\textbf{b}_{m,f})$ and then decide the cache strategy directly.

\section{Cache Strategy in Stationary Environment}
In this section, we aim to solve $\textbf{P2}$ in the stationary environment. We first introduce the stationary environment and give the specific form of $\textbf{P2}$. Then we solve this problem with MAMAB-based methods in both distributed and collaborative manners. Finally, we modify the MAMAB algorithms by designing new perturbed terms to achieve a better performance.
\subsection{Stationary Environment}
The stationary environment is described as follows. The file library $\mathcal{F}$ is a static set with fixed size $F$ and user preference is time-invariant. Denote $p_{u,f}$ as the probability that user $u$ requests file $f$, which satisfies $0\leq p_{u,f}\leq 1$ and $\sum_{f=1}^Fp_{u,f}=1$. Each user is assumed to request one file at each time slot independently according to its own user preference. Moreover, the locations of users are fixed. Thus, transmission delays $d_{m,u}$ are time-invariant and can be observed when SBSs serve users or computed if the communication distance is known.

In this stationary environment, the expected reward of SBS $m$ by caching file $f$ is given by:
\begin{align}
r_{m,f}(\textbf{b}_{m,f})&=a_{m,f}\sum_{u\in\mathcal{U}_m}p_{u,f} \left(d_0-d_{m,u}\right) \nonumber\\
&~~~~~~~~~~\sum_{j=1}^{|\mathcal{M}_u|}\mathbb{I}\{m=j_u\}\prod_{n=1}^{j-1}(1-a_{n_u,f}).  \label{eqn:reward2}
\end{align}

Since user preference is time-invariant, the solutions of different SPO problems of $\textbf{P2}$ are identical. Hence, we only need to focus on one SPO problem, which is given by:
\addtocounter{equation}{1}
\begin{align}
\text{\textbf{P3:}} \quad \underset{\mathbf{A}}{\text{max}} \quad & R_{\text{total}}(\textbf{A})=\sum_{m=1}^M\sum_{f=1}^Fr_{m,f}(\textbf{b}_{m,f}) \tag{\theequation a}   \label{eqn:reward} \\
\text{s.t} \quad &\sum_{f=1}^F a_{m,f} \leq S, ~~~  m=1,2,\ldots,M  \tag{\theequation b}  \\
&a_{m,f}\in\{0,1\}, ~~~  m=1,2,\ldots,M.  \tag{\theequation c}
\end{align}
%

Then we aim to solve the SPO problem $\textbf{P3}$ by using MAMAB-based algorithms in both distributed and collaborative manners with different reward decomposition methods.
\subsection{Agent-based Collaborative MAMAB}
In the following, we propose two agent-based collaborative MAMAB algorithms with different reward decomposition methods from the SBS perspective and user perspective, respectively.
\subsubsection{SBS perspective}
From (\ref{eqn:reward}), it is observed that the total reward function $R_{\text{total}}(\textbf{A})$ is decomposed into a linear combination of local agent-dependent {reward functions} $r_{m,f}(\textbf{b}_{m,f})$. The reward that SBS $m$ can obtain depends on the joint $\textbf{b}_{m,f}\in\{0,1\}^{1\times(\Gamma(m)+1)}$. In the initial phase $t=0$, {SBSs cache files randomly to make each joint action $\textbf{b}_{m,f}$ appear at least once. Note that $t=0$ can contain multiple time slots actually.} The number of times that $\textbf{b}_{m,f}$ appears until $t$ is denoted as $T_{m,f}^t(\textbf{b}_{m,f})$. Besides, if $\textbf{b}_{m,f}$ appears at time slot $t$, we update the average reward $\overline{R}_{m,f}^t(\textbf{b}_{m,f})$ of SBS $m$ with joint action $\textbf{b}_{m,f}$ as:
\begin{align}
\overline{R}_{m,f}^t(\textbf{b}_{m,f})=\frac{\left[T_{m,f}^t(\textbf{b}_{m,f})-1\right]\overline{R}_{m,f}^{t-1}(\textbf{b}_{m,f})+r_{m,f}^t(\textbf{b}_{m,f})}{T_{m,f}^t(\textbf{b}_{m,f})}, \label{eqn:ave_agent1}
\end{align}
where $r_{m,f}^t(\text{b}_{m,f})$ is the observed reward of SBS $m$ with joint action $\textbf{b}_{m,f}$ at time slot $t$. Each SBS $m$ needs to share its cache strategy with its neighbor $\Gamma(m)$ in the coordination graph in order to update $T_{m,f}^t(\textbf{b}_{m,f})$ and $\overline{R}_{m,f}^t(\textbf{b}_{m,f})$. Note that the average reward $\overline{R}_{m,f}^t(\textbf{b}_{m,f})$ is inaccurate when $T_{m,f}^t(\textbf{b}_{m,f})$ is small and hence utilizing the the average reward to decide the cache strategy is not appropriate. Therefore, we define the estimated reward by adding a perturbed term to the average reward similar to \cite{chen2013combinatorial}, which can achieve a good tradeoff between exploration and exploitation. The estimated reward $\widehat{R}_{m,f}^t(\textbf{b}_{m,f})$ is updated as:
\begin{align}
\widehat{R}_{m,f}^t(\textbf{b}_{m,f})=\overline{R}_{m,f}^{t-1}(\textbf{b}_{m,f})+B_{\text{a-coll},m}(\textbf{b}_{m,f})\sqrt{\frac{3\log(t)}{2T^{t-1}_{m,f}(\textbf{b}_{m,f})}}, \label{eqn:est_agent1}
\end{align}
where $B_{\text{a-coll},m}(\textbf{b}_{m,f})$ is the upper bound on the reward of SBS $m$ with joint action $\textbf{b}_{m,f}$ as:
\begin{align}
B_{\text{a-coll},m}(\textbf{b}_{m,f})&=a_{m,f}\sum_{u\in\mathcal{U}_m}(d_0-d_{m,u}) \nonumber\\
&~~~~~~~~~\sum_{j=1}^{|\mathcal{M}_u|}\mathbb{I}\{m=j_u\}\prod_{n=1}^{j-1}(1-a_{n_u,f}).
\end{align}

{Note that this perturbed term corresponds to the upper confidence bound (UCB) in MAB, which is utilized in the combinatorial MAB problem [30] with a single SBS. This UCB term can also be utilized in this MAMAB problem since the reward of each SBS for a fixed joint action has a stationary distribution in the stationary environment and the theoretical proof is shown later in Lemma $2$.}

With the estimated reward, we can optimize the cache strategy by maximizing the total estimated reward $\sum_{m=1}^M\sum_{f=1}^F\widehat{R}_{m,f}^t(\textbf{b}_{m,f})$ at each time slot $t$ with the cache size constraint. Note that this problem is NP-complete and the commonly used greedy algorithm has a high computational complexity\cite{shanmugam2013femtocaching}. Therefore, we propose a low-complexity coordinate ascent algorithm \cite{vlassis2004anytime} to obtain the solutions. Coordinate ascent algorithm is an anytime algorithm that is appropriate for real-time multi-agent systems where decision making must be done under time constraints. Given a random initial cache strategy, each SBS optimizes its own cache {and share its optimized cache strategy with its neighbor SBSs in the coordination graph} sequentially while the other $M-1$ SBSs stay the same. This step is repeated until no improvement can be made under the time constraints. The details of the coordinate ascent algorithm {are} given in Algorithm \ref{alg:9}.
\begin{algorithm}
\caption{Coordinate Ascent Algorithm} \label{alg:9}
\begin{algorithmic}[1]
\STATE Initialize the cache of SBSs randomly
\REPEAT
\FOR {$m=1,2,\ldots,M$}
\STATE Calculate the estimated reward of caching file $f$ at SBS $m$ for all files when the caching actions of the other $M-1$ SBSs stay the same
\STATE SBS $m$ optimizes its own caching action by maximizing the total estimated reward
\ENDFOR
\UNTIL {No SBS changes its cache strategy or decision time runs out}
\end{algorithmic}
\end{algorithm}

For the performance of Algorithm \ref{alg:9}, we have the following lemma.
\begin{lemma}
Algorithm \ref{alg:9} can achieve at least $\frac{1}{2}$ optimality of \textbf{P3}, i.e., $R_{\text{total}}(\textbf{A}^L)\geq\frac{1}{2}R_{\text{total}}(\textbf{A}^*)$, where $\textbf{A}^L$ is the cache strategy obtained by Algorithm \ref{alg:9}.
\end{lemma}
\begin{proof}
See Appendix A.
\end{proof}

Note that SBS $m$ with joint action $\textbf{b}_{m,f}$ only can obtain the reward for file $f$ if $a_{m,f}=1$, the number of effective action space is $\sum_{m=1}^M2^{\Gamma(m)}F$, which grows exponentially with $\Gamma(m)$. Therefore, this agent-based collaborative MAMAB in a SBS perspective has a high computational complexity and a slow learning speed.
\subsubsection{User perspective}
{Note that from the SBS perspective}, the action space of $\textbf{b}_{m,f}$ is fixed. However, the reward distributions of SBS $m$ for some different $\textbf{b}_{m,f}$ are independent and identical and can be combined to reduce the action space. For example, if there is only one user $u$ covered by two SBSs, the reward distributions of SBS {$1_u$ with joint action $(a_{1_u,f}=1,a_{2_u,f}=1)$ and $(a_{1_u,f}=1,a_{2_u,f}=0)$ are identical since $u$ will always be served by SBS $1_u$ if $a_{1_u,f}=1$.} Therefore, we focus on the reward of SBSs {from the} user perspective by decomposing total reward function into a linear combination of local agent-dependent reward {functions} $r_{m,f}(\textbf{c}_{m,\mathcal{V}_m,f})$ as:
\begin{align}
R_{\text{total}}(\textbf{A})=\sum_{m=1}^M\sum_{f=1}^F\sum_{\mathcal{V}_m} r_{m,f}(\textbf{c}_{m,\mathcal{V}_m,f})
\end{align}
where $\textbf{c}_{m,\mathcal{V}_m,f}=[a_{m,f}=1,a_{n_1,f}=0,\ldots,a_{n_{|\mathcal{V}_m|,f}=0}]$ is the joint action that SBS $m$ caches file $f$ while SBSs in $\mathcal{V}_m$ {do not} and $\mathcal{V}_m\in\{\mathcal{V}^u_m\mid u\in\mathcal{U}_m\}$, where $\mathcal{V}^u_m\triangleq\{n\in\mathcal{M}_u\mid d_{n,u}<d_{m,u}\}$ is the set of SBSs that are closer to user $u$ than SBS $m$. The reason for only considering joint action $\textbf{c}_{m,\mathcal{V}_m,f}$ is that SBS $m$ can serve user $u$ if and only if it is the nearest SBS that caches the requested file of user $u$.
In the initial phase $t=0$, {SBSs cache files randomly to make each joint action $\textbf{c}_{m,\mathcal{V}_m,f}$ appear at least once similar to the agent-based collaborative MAMAB from the SBS perspective.} The number of times that $\textbf{c}_{m,\mathcal{V}_m,f}$ appears until $t$ is denoted as $T^t_{m,f}(\textbf{c}_{m,\mathcal{V}_m,f})$. Besides, if joint action $\textbf{c}_{m,\mathcal{V}_m,f}$ appears at time slot $t$, we update the average reward $\overline{R}_{m,\mathcal{V}_m,f}^t(\textbf{c}_{m,\mathcal{V}_m,f})$ as:
\begin{align}
\!\!\overline{R}_{m,\mathcal{V}_m,f}^t(\textbf{c}_{m,\mathcal{V}_m,f})&\!=\!\frac{\left[T^t_{m,\mathcal{V}_m,f}(\textbf{c}_{m,\mathcal{V}_m,f})\!-\!1\right]\!\overline{R}_{m,\mathcal{V}_m,f}^{t-1}(\textbf{c}_{m,\mathcal{V}_m,f})}{T^t_{m,\mathcal{V}_m,f}(\textbf{c}_{m,\mathcal{V}_m,f})} \nonumber\\
&~~~~~+\frac{\sum_{u\in \mathcal{U}_m}\mathbb{I}\{\mathcal{V}^{u}_m=\mathcal{V}_m\}r^{u,t}_{m,f}}{T^t_{m,\mathcal{V}_m,f}(\textbf{c}_{m,\mathcal{V}_m,f})},
\end{align}
where $r_{m,f}^{u,t}$ is the reward of SBS $m$ obtained by serving user $u$ with file $f$. {Note that each SBS $m$ needs to share its cache strategy with its neighbor $\Gamma(m)$ in the coordination graph in order to update $T^t_{m,f}(\textbf{c}_{m,\mathcal{V}_m,f})$ and average reward $\overline{R}_{m,\mathcal{V}_m,f}^t(\textbf{c}_{m,\mathcal{V}_m,f})$.} By adding a perturbed term similar to (\ref{eqn:est_agent1}), the estimated reward with joint action $\textbf{c}_{m,\mathcal{V}_m,f}$ at time slot t is given by:
\begin{align}
&\widehat{R}_{m,\mathcal{V}_m,f}^t(\textbf{c}_{m,\mathcal{V}_m,f})=\overline{R}_{m,\mathcal{V}_m,f}^{t-1}(\textbf{c}_{m,\mathcal{V}_m,f})  \nonumber\\
&~~~~+B_{\text{a-coll},m,\mathcal{V}_m}(\textbf{c}_{m,\mathcal{V}_m,f})\sqrt{\frac{3\log(t)}{2T^{t-1}_{m,\mathcal{V}_m,f}(\textbf{c}_{m,\mathcal{V}_m,f})}},
\end{align}
where $B_{\text{a-coll},m,\mathcal{V}_m}(\textbf{c}_{m,\mathcal{V}_m,f})$ is the upper bound on the reward of SBS $m$ with $\textbf{c}_{m,\mathcal{V}_m,f}$ as:
\begin{align}
B_{\text{a-coll},m,\mathcal{V}_m}(\textbf{c}_{m,\mathcal{V}_m,f})&=a_{m,f}\sum_{u\in\mathcal{U}_m}(d_0-d_{m,u}) \nonumber\\
&\mathbb{I}\{\mathcal{V}^u_m=\mathcal{V}_m\}\prod_{n\in\mathcal{V}_m}(1-a_{n,f}).
\end{align}

SBSs optimize the cache strategies to maximize the total estimated reward $\sum_{f=1}^F\sum_{m=1}^M\sum_{\mathcal{V}_m}\widehat{R}_{m,\mathcal{V}_m,f}^t(\textbf{c}_{m,\mathcal{V}_m,f})$ at each time slot $t$ by utilizing Algorithm \ref{alg:9}. Note that the effective action space of $\textbf{c}_{m,\mathcal{V}_m,f}$ depends on the number of users and is $\sum_{m=1}^M2^{\Gamma(m)}F$ in the worst case. {Therefore, we conclude that commonly used agent-based collaborative MAMAB from the SBS perspective can be seen as the worst case of this new agent-based collaborative MAMAB algorithm from the user perspective. By analyzing the property of the caching problem, our proposed algorithm from the user perspective decreases the action space greatly compared to that from the SBS perspective when the user number is small by combining some joint actions together.} However, as the user number keeps growing, the action space can be still very large as it grows exponentially with $\Gamma(m)$, which causes a high computational complexity and a slower learning speed.

\begin{lemma}
The $(\alpha,\beta)$-approximation regret of these two agent-based collaborative MAMAB algorithms using an $(\alpha,\beta)$-approximation is on the order of $\mathcal{O}(\log{(T_{\text{total}})})$.
\end{lemma}
\begin{proof}
See Appendix B.
\end{proof}

{This regret bound guarantees that the algorithm has an asymptotically optimal performance since $\lim_{T_{\text{total}}\rightarrow \infty}\frac{\log {(T_{\text{total}}})}{T_{\text{total}}}=0$. Therefore, we conclude that these two agent-based collaborative MAMAB algorithms converge to the optimal offline cache strategy in the stationary environment.}
\subsection{Distributed MAMAB}
To avoid the curse of dimensionality, we try to solve $\textbf{P3}$ in a distributed manner. In this case, each SBS is regarded as an independent learner and it learns its own cache strategy independently without information exchange with other SBSs. Therefore, the total reward function is decomposed into a linear combination of agent-independent reward {functions} $r_{m,f}(a_{m,f})$ as:
\begin{align}
R_{\text{total}}(\textbf{A})=\sum_{m=1}^M\sum_{f=1}^Fr_{m,f}(a_{m,f}),  \label{eqn:dis_reward}
\end{align}
which means each SBS totally ignores the reward and action of other SBSs. {Note that (\ref{eqn:dis_reward}) is not equivalent to (\ref{eqn:reward}). It is just a function decomposition method, which is called the independent learner approach \cite{claus1998dynamics,kok2006collaborative} and is commonly used in reinforcement learning. In independent learner approach, the agents ignore the actions and rewards of the other agents, and learn their strategies independently.}
In the initial phase $t=0$, each SBS independently caches all files once. The number of times that $a_{m,f}$ appears until $t$ is denoted as $T^t_{m,f}(a_{m,f})$. Besides, if $a_{m,f}$ appears at time slot $t$, the average reward $\overline{R}_{m,f}^t(a_{m,f})$ of SBS $m$ with caching action $a_{m,f}$ is updated as:
\begin{align}
\overline{R}_{m,f}^t(a_{m,f})=\frac{\left[T^t_{m,f}(a_{m,f})-1\right]\overline{R}^{t-1}_{m,f}(a_{m,f})+r^t_{m,f}(a_{m,f})}{T^t_{m,f}(a_{m,f})}, \label{eqn:ave_dis}
\end{align}
where $r^t_{m,f}(a_{m,f})$ is the observed reward of SBS $m$ with action $a_{m,f}$ at time slot $t$. By adding a perturbed term similar to (\ref{eqn:est_agent1}), the estimated reward with $a_{m,f}$ at time slot $t$ is given by:
\begin{align}
\widehat{R}_{m,f}^t(a_{m,f})=\overline{R}_{m,f}^{t-1}(a_{m,f})+B_{d,m}(a_{m,f})\sqrt{\frac{3\log(t)}{2T^{t-1}_{m,f}(a_{m,f})}}. \label{eqn:est_dis}
\end{align}
where $B_{d,m}(a_{m,f})$ is the upper bound on the reward of SBS $m$ with action $a_{m,f}$ and is given by:
\begin{align}
B_{d,m}(a_{m,f})=a_{m,f}\sum_{u\in\mathcal{U}_m}(d_0-d_{m,u})
\end{align}


Each SBS optimizes its cache strategy by maximizing its own estimated reward $\sum_{f=1}^F\widehat{R}^t_{m,f}(a_{m,f})$ at each time slot $t$ independently. In this case, the effective action space is $MF$ and has a linear growth speed of the number of SBSs, which has a low computational complexity and a fast learning speed. However, the performance is poor since there is no coordination among SBSs.

\subsection{Edge-based Collaborative MAMAB}
Based on the analysis of the above three MAMAB approaches, it is clear that a good solution should consider both SBS coordination and computational complexity simultaneously. Therefore, we propose a coordination graph edge-based reward assignment method by assigning the reward of SBSs to the edges in the coordination graph $(V,E)$. In this case, we decompose the total reward function into a linear combination of edge-independent reward {functions} as:
\begin{align}
R_{\text{total}}(\textbf{A})=\sum_{f=1}^F\sum_{(m,n)\in E} r_{m,n,f}(a_{m,f},a_{n,f}). \label{eqn:reward3}
\end{align}
where $r_{m,n,f}(a_{m,f},a_{n,f})$ is the reward of the edge $(m,n)$ with joint action $(a_{m,f},a_{n,f})$. {Note that (\ref{eqn:reward3}) and (\ref{eqn:reward}) are not equivalent. It is just a reward assignment method by dividing the reward of agents to the edges in the coordination graph, which is utilized in reinforcement learning to reduce the complexity \cite{kok2006collaborative}. The main difficulty is how to design the specific coordination graph edge-based reward assignment method based on the problem property.}

The coordination graph edge-based reward assignment method runs as follows. For a user $u$ requesting file $f$ at time slot $t$, we assume that it is served by its neighbor SBS $m$ and SBS $m$ can obtain a reward $r_{m,f}^{u,t}$. Note that SBS $m$ knows the locations of other $M-1$ SBSs and its serving user $u$. If SBS $m$ is the nearest SBS of $u$, then the reward $r_{m,f}^{u,t}$ obtained by serving $u$ will be entirely assigned to its self edge $(m,m)$. This is because $u$ will always be served by $m$ in this case, regardless of whether other SBSs cache file $f$ or not. When SBS $m$ is not the nearest SBS of $u$, it serves $u$ if and only if the nearer neighbor SBSs of $u$ do not cache the requested file $f$. In this case, the reward $r_{m,f}^{u,t}$ is divided equally among the edges between the serving SBS $m$ and all the nearer neighbor SBSs $n\in\mathcal{V}^u_m$ with joint action $(a_{m,f}=1,a_{n,f}=0)$ since the reward $r_{m,f}^{u,t}$ depends on cache actions of SBS $m$ and $n\in\mathcal{V}^u_m$. These split rewards on edges $(m,n)$ are stored at the serving SBS $m$ since only SBS $m$ caches the requested file $f$ and obtains the reward while other nearer SBSs $n\in\mathcal{V}^u_m$ not.

{SBSs cache files randomly to make each joint action $(a_{m,f},a_{n,f})$ for $(m,n)\in E$ appear at least once in the initial phase $t=0$.\footnote{Note that the initial phase in the edge-based collaborative MAMAB algorithm is much shorter than the agent-based ones due to its smaller action space.}} At each time slot $t$, SBSs update the number of times that joint action $(a_{m,f},a_{n,f})$ appears until $t$, which is denoted as $T^t_{m,n,f}(a_{m,f},a_{n,f})$. If $(a_{m,f},a_{n,f})$ appears at time slot $t$, the average reward $\overline{R}_{m,n,f}^t(a_{m,f},a_{n,f})$ of the edge $(m,n)\in E$ with action pair $(a_{m,f},a_{n,f})$ is updated as:
\begin{align}
&\overline{R}_{m,n,f}^t(a_{m,f},a_{n,f})=\nonumber\\
&\frac{\left[T^t_{m,n,f}(a_{m,f},a_{n,f})-1\right]\overline{R}_{m,n,f}^{t-1}(a_{m,f},a_{n,f})+r_{m,n,f}^t(a_{m,f},a_{n,f})}{T^t_{m,n,f}(a_{m,f},a_{n,f})},  \label{eqn:ave_col}
\end{align}
where $r_{m,n,f}^t(a_{m,f},a_{n,f})$ is the reward assigned to the edge $(m,n)$ with $(a_{m,f},a_{n,f})$ at time slot $t$. {Note that each SBS $m$ needs to share its cache strategy with its neighbor $\Gamma(m)$ in the coordination graph in order to update $T^t_{m,n,f}(a_{m,f},a_{n,f})$ and average reward $\overline{R}_{m,n,f}^t(a_{m,f},a_{n,f})$. Moreover, each SBS $m$ only needs to store the reward of edge $(m,n)$ with joint action $(a_{m,f}=1,a_{n,f}=0)$ for $n\in\Gamma(m)$ and its self edge $(m,m)$ with action $(a_{m,f}=1,a_{m,f}=1)$ due to the coordination graph edge-based reward assignment method.} By adding a perturbed term similarly, the estimated reward $\widehat{R}_{m,n,f}^t(a_{m,f},a_{n,f})$ of the edge $(m,n)$ with joint action $(a_{m,f},a_{n,f})$ at time slot $t$ is given by:
\begin{align}
\!\!\!\!&\widehat{R}_{m,n,f}^t(a_{m,f},a_{n,f})=\overline{R}_{m,n,f}^{t-1}(a_{m,f},a_{n,f}) \nonumber\\
&~~~~~~~~+B_{\text{coll},m,n}(a_{m,f},a_{n,f})\sqrt{\frac{3\log(t)}{2T^{t-1}_{m,n,f}(a_{m,f},a_{n,f})}}, \label{eqn:est_col}
\end{align}
where $B_{\text{coll},m,n}(a_{m,f},a_{n,f})$ is the upper bound of the reward on the edge $(m,n)\in E$ with joint action $(a_{m,f},a_{n,f})$ and it is given by:
\begin{align}
&B_{\text{coll},m,n}(a_{m,f},a_{n,f})= \nonumber\\
&\begin{cases}
a_{m,f}\sum_{u\in\mathcal{U}_m}\mathbb{I}\{m=1_{u}\}(d_0-d_{m,u}),  ~~~~~~~~\text{if}~ m=n   \\
a_{m,f}(1-a_{n,f})\sum_{u\in\mathcal{U}_m}\mathbb{I}\{n\in\mathcal{V}_{m}^u \}\sum_{j=2}^{|\mathcal{M}_u|} \\
\mathbb{I}\{m=j_u\}\frac{d_0-d_{m,u}}{j-1}+a_{n,f}(1-a_{m,f})\sum_{u\in\mathcal{U}_n}  \\
\mathbb{I}\{m\in\mathcal{V}_{n}^u \}\sum_{j=2}^{|\mathcal{M}_u|}\mathbb{I}\{n=j_u\}\frac{d_0-d_{n,u}}{j-1},  ~~~~~~~\text{if} ~m\neq n
\end{cases}.
\end{align}

\begin{algorithm}
\caption{Edge-based Collaborative MAMAB Caching}  \label{alg:3}
\begin{algorithmic}[1]
\STATE Cache files at SBSs to make all joint actions $(a_{m,f},a_{n,f})$ appear once for $(m,n)\in E$, i.e., $T^0_{m,n,f}(a_{m,f},a_{n,f})=1$ and then update the average reward $\overline{R}^0_{m,n,f}(a_{m,f},a_{n,f})$
\FOR {$t=1,2,\ldots,T_{\text{total}}$}
\STATE Each SBS decides its cache strategy $\{\widetilde{a}_{m,f}\}$ with the distributed MAMAB algorithm
\STATE Update the estimated reward according to (\ref{eqn:est_col}) for all joint actions $(a_{m,f},a_{n,f})$
\STATE Decide the cache strategy by maximizing the total estimated reward $\sum_{f=1}^F\sum_{(m,n)\in E}$ $\widehat{R}_{m,n,f}^t(a_{m,f},a_{n,f})$ according to the Algorithm \ref{alg:9} with the initial cache strategy $\{\widetilde{a}_{m,f}^t\}$
\STATE Initialize the reward $r^t_{m,n,f}(a_{m,f},a_{n,f})=0$ for all edges $(m,n)\in E$ and joint actions
\STATE Count $T^t_{m,n,f}(a_{m,f},a_{n,f})$ for all joint actions
\FOR {$u=1,2,\ldots,U$}
\STATE Observe the reward $r_{m,f}^{u,t}$ of the SBS $m$ that serves user $u$ with file $f$
\IF{SBS $m$ is the nearest SBS of user $u$}
\STATE  Update $r_{m,m,f}^t(a_{m,f},a_{m,f})=r_{m,m,f}^t(a_{m,f},a_{m,f})+r_{m,f}^{u,t}$
\ELSE
\STATE Update $r_{m,n,f}^t(a_{m,f},a_{n,f})=r_{m,n,f}^t(a_{m,f},a_{n,f})+\frac{r_{m,f}^{u,t}}{j-1}$ for the edges $(m,n)$ between $j-1$ nearer SBSs and $j$-th nearest SBS $m$ of user $u$
\ENDIF
\ENDFOR
\STATE Update the average reward according to (\ref{eqn:ave_col}) if joint action $(a_{m,f},a_{n,f})$ appears and $\overline{R}_{m,n,f}^t(a_{m,f},a_{n,f})=\overline{R}_{m,n,f}^{t-1}(a_{m,f},a_{n,f})$ otherwise.
\ENDFOR
\end{algorithmic}
\end{algorithm}

SBSs optimize the cache strategies to maximize the total estimated reward $\sum_{f=1}^F\sum_{(m,n)\in E}$ $\widehat{R}_{m,n,f}^t(a_{m,f},a_{n,f})$ at each time slot $t$ by utilizing Algorithm \ref{alg:9}. Specifically, we choose cache strategy obtained from the distributed MAMAB algorithm as the initial cache strategy in Algorithm \ref{alg:9} to have a better performance and a faster convergence speed.
Note that the edge-based collaborative MAMAB algorithm not only considers the SBS coordination, but also reduces the computational complexity greatly {with} effective action space $\sum_{m=1}^M\Gamma(m)+M$, which grows linearly with the number of edges. The details of this algorithm {are} given in Algorithm \ref{alg:3}.

%
%


\subsection{Tradeoff between Exploration and Exploitation}
As we have discussed, we add a perturbed term to the average reward to achieve the tradeoff between exploration and exploitation. However, there also {exists} some simple methods that can achieve a good tradeoff in MAMAB problems. One such simple and widely used algorithm is the $\epsilon$-greedy algorithm, which caches files to maximize the total average reward with probability $1-\epsilon$ and caches files randomly with probability $\epsilon$. $\epsilon$-greedy algorithm can also be divided into different manners similarly with different reward decomposition methods.

We also propose a new perturbed term to update the estimated reward by exploiting the particular structure of the problem similar to \cite{blasco2014learning} and the new perturbed term is $\sqrt{\frac{3\log(B^2t)}{2T^{t-1}}}$ rather than $B\sqrt{\frac{3\log(t)}{2T^{t-1}}}$ if the reward upper bound $B>0$. Otherwise, the perturbed term is $0$. The MAMAB algorithms with new perturbed terms are called MAMAB-v2 algorithms.

\section{Cache Strategy in Non-Stationary Environment}
In this section, we aim to solve $\textbf{P2}$ in the non-stationary environment. We first give a brief introduction of the non-stationary environment. {Then we modify our proposed MAMAB algorithms to make them adapt to the non-stationary environment. Finally, we discuss the tradeoff between exploration and exploitation.}
\subsection{Non-stationary Environment}
The non-stationary environment is described as follows. The file library $\mathcal{F}$ is a dynamic set with new files entering and old files leaving over time. Besides, user preference changes irregularly over time and the number of requests for each user at a time slot is random. {Note that users can keep silent and request no file at some time slots, the set of active users that request files also changes over time. Moreover, we take user mobility into account that user locations can change at different time slots. Thus, transmission delay $d_{m,u}$ also varies over time and cannot be known in advance. This non-stationary environment is more realistic since the real world scenario is actually non-stationary. However, this non-stationary environment also brings some challenges since it is time-varying and more stochastic. Therefore, we need to modify the algorithms proposed in the stationary environment to adapt to this non-stationary environment.\footnote{Nota that we only modify the distributed and edge-based collaborative MAMAB in this work since the action space of the agent-based collaborative MAMAB is too large, which is almost impractical to be utilized in the non-stationary environment.}}

\subsection{Modified MAMAB}
{Note that the proposed MAMAB algorithms in stationary environment have some prior information of the environment, such as static file set. However, SBSs have no knowledge of these information in the non-stationary environment. Therefore, we need to define an active file set $\mathcal{F}_{\text{active}}^t$ at each time slot $t$ due to the fact that the file library $\mathcal{F}$ is dynamic over time. $\mathcal{F}_{\text{active}}^t$ is the set of files that are requested from time slot $1$ to $t$ and needs to be updated at the end of each time slot $t$ based on the receiving requests. Besides, we assign a large enough initial value to all joint actions of new added files in $\mathcal{F}_{\text{active}}^t$, rather than caching files randomly, to make each new joint action appear at least once. This new initialization method can reduce the explorations in the initial phase greatly and improve the performance compared to the initial phase in the stationary environment. Moreover, we need to design new UCB terms accordingly since SBSs have no information of the reward upper bound $B$.}

\subsubsection{Modified Distributed MAMAB}
{From the distributed MAMAB algorithm in the stationary environment, it is observed that the average reward $\overline{R}_{m,f}^t(a_{m,f})=0$ if $a_{m,f}=0$. Therefore, each SBS $m$ only initializes the estimated reward $\widehat{R}_{m,f}^t(a_{m,f})=H$ when $a_{m,f}=1$ for all active files $f\in\mathcal{F}_{\text{active}}^0$, where $H$ is a large enough number in order to make each active file can be cached at SBSs at least once. At each time slot $t$, each SBS $m$ updates the estimated reward for all files that have been cached at least once as:
\begin{align}
\widehat{R}_{m,f}^t(a_{m,f})=\overline{R}_{m,f}^{t-1}(a_{m,f})+\sqrt{\frac{3\log\left((B^t_{d,m})^2t\right)}{2T^{t-1}_{m,f}(a_{m,f})}},
\end{align}
where $B^t_{d,m}=\underset{f\in\mathcal{F}_{\text{active}}^t}{\text{max}}\overline{R}_{m,f}^{t-1}(a_{m,f})$. Comparing with the distributed MAMAB in the stationary environment, we utilize the maximal average reward to represent the upper bound of the reward due to the upper bound of the reward cannot be obtained and changes over time in the non-stationary environment. Then each SBS decides its cache strategy by maximizing its own estimated reward. After caching action has been done, each SBS $m$ updates the number of times that $a_{m,f}$ appears $T^t_{m,f}(a_{m,f})$ and the average reward according to (\ref{eqn:ave_dis}) based on the observed reward similar to the distributed MAMAB in the stationary environment. Finally, SBSs update the active file set $\mathcal{F}_{\text{active}}^t$ at the end of time slot $t$ based on their receiving requests and initialize the estimated reward $\widehat{R}_{m,f}^t(a_{m,f})=H$ when $a_{m,f}=1$ for new added files.}

\subsubsection{Modified Edge-based Collaborative MAMAB}
{
From Algorithm \ref{alg:3}, it is observed that the average reward $\overline{R}_{m,n,f}^t(a_{m,f},a_{n,f})=0$ if action pair $(a_{m,f},a_{n,f})=(1,1)$ when $m\neq n$ or $(a_{m,f},a_{n,f})=(0,0)$. Therefore, we only need to focus on the edge $(m,n)\in E$ when only one of them caches the file and the other does not, i.e., $(a_{m,f},a_{n,f})=(1,0)$, $(a_{m,f},a_{n,f})=(0,1)$ and $(a_{m,f},a_{m,f})=(1,1)$, which are called effective joint actions. We initialize the estimated reward of all effective joint actions $\widehat{R}_{m,n,f}^t(a_{m,f},a_{n,f})$ with a large enough number $H$. At time slot $t$, SBSs update the estimated reward of all effective joint actions $(a_{m,f},a_{n,f})$ that appear at least once as:
\begin{align}
\widehat{R}_{m,n,f}^t(a_{m,f},a_{n,f})&=\overline{R}_{m,n,f}^{t-1}(a_{m,f},a_{n,f}) \nonumber\\
&~~~~~+\sqrt{\frac{3\log((B^t_{\text{coll},m,n})^2t)}{2T^{t-1}_{m,n,f}(a_{m,f},a_{n,f})}},
\end{align}
where $B^t_{\text{coll},m,n}=\underset{f\in\mathcal{F}_{\text{active}}^t}{\text{max}}\overline{R}_{m,n,f}^{t-1}(a_{m,f},a_{n,f})$ is seen as the upper bound of the reward on the edge $(m,n)\in E$ with joint action $(a_{m,f},a_{n,f})$ for all active files $f$ at time slot $t$. Then SBSs decide their cache strategies collaboratively to maximize the total estimated reward according to Algorithm \ref{alg:9}. SBSs update $T^t_{m,n,f}(a_{m,f},a_{n,f})$ and average reward $\overline{R}_{m,n,f}^t(a_{m,f},a_{n,f})$ according to (\ref{eqn:ave_col}) if joint action $(a_{m,f},a_{n,f})$ appears based on the observed reward that follow the coordination graph edge-based reward assignment. Finally, SBSs update the active file set $\mathcal{F}_{\text{active}}^t$ at the end of time slot $t$ based on their receiving requests and initialize the estimated reward $\widehat{R}_{m,n,f}^t(a_{m,f},a_{n,f})=H$ for all effective joint actions of new added files.}

{
\subsection{Tradeoff between Exploration and Exploitation}
As we have discussed, by adding the perturbed term to the average reward as the estimated reward is a good way to achieve the tradeoff between exploration and exploitation in the stationary environment. However, our proposed perturbed terms do not have theoretical performance guarantee in the non-stationary environment. Therefor, we can also utilize the $\epsilon$-greedy algorithm to achieve the tradeoff between exploration and exploitation, which is simple but effective in many scenarios. In the $\epsilon$-greedy algorithm, SBSs cache files to maximize the total average reward with probability $1-\epsilon$ and caches files randomly with probability $\epsilon$. Note that $\epsilon$-greedy algorithm can also be divided into both distributed and collaborative manners similarly with different average reward update manners.
}
\section{Simulation Results}
In this section, we demonstrate the performance of {our} proposed algorithms in both stationary environment and non-stationary environment. Simulations are performed in a square area of $100\times100$ $\text{m}^2$. Both SBSs and users are uniformly distributed in this plane. Unless otherwise stated, the system parameters are set as follows: bandwidth $W=10~\text{MHz}$, transmission power of SBSs $P=1$ $\text{W}$, Gaussian white noise power $\sigma_N^2=1$ $\text{W}$, path loss exponent $\alpha=4$ and the core network transmission delay {$d_0=3\underset{m\in\mathcal{M},u\in\mathcal{U}}{\text{max}}~\frac{1}{W\log_2(1+\text{SNR}_{m,u})}$}.

\subsection{Stationary Environment}
In the stationary environment, we assume that users request for files according to independent Zipf distributions. Specifically, the probability of user $u$ requesting file $f$ is denoted as $p_{u,f}=\frac{1/f_u^{\delta_u}}{\sum_{j=1}^F1/j^{\delta_u}}$, where $\delta_u\in \{0.5,0.7,0.9,1.1,1.3\}$ is the Zipf parameter of user $u$ and $f_u$ means that file $f$ is the $f_u$-th most interested file of $u$. Unless otherwise stated, the simulation parameters are set as follows: the number of files $F=100$, cache size $S=10$ and time horizon $T_{\text{total}}=25000$. We utilize the average transmission delay as the performance metric, which is defined as the ratio of the total transmission delay to the the number of time slots.

\begin{figure*}[t]
\vspace{-0.3cm}
\begin{subfigure}[t]{0.5\linewidth}
\centering
\includegraphics[width=8.5cm]{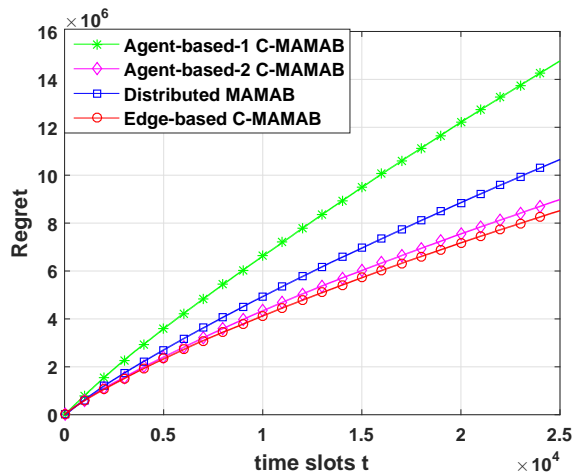}\\
\caption{{Regret comparison of MAMAB algorithms.}} \label{fig:M=6_v1}
\end{subfigure}
\hfill
\begin{subfigure}[t]{0.5\linewidth}
\centering
\includegraphics[width=8.5cm]{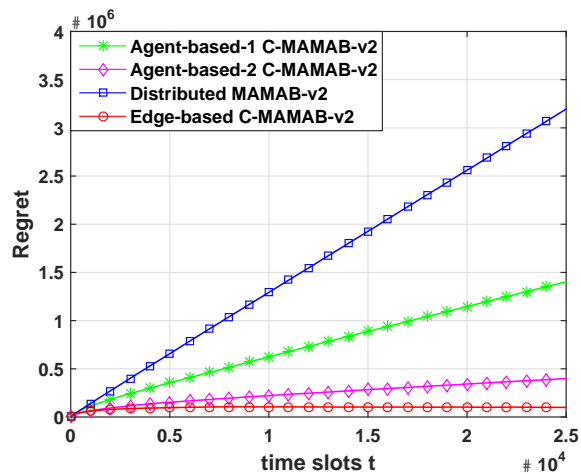}\\
\caption{{Regret comparison of MAMAB-V2 algorithms.}} \label{fig:M=6_v3}
\end{subfigure}
\caption{{The time evolution of regret when users have the same preference with $M=6$, $U=50$, $l_c=50$ \text{m}, $F=100$, $S=10$ and $\delta=0.9$.}}
\vspace{-0.4cm}
\end{figure*}

To evaluate the performance of the algorithms proposed in this paper, we compare them with the following benchmark algorithms:
%
%

$\bullet$ \textbf{Oracle Coordinate Ascent Algorithm}: Cache files according to Algorithm \ref{alg:9} with $300$ random initial strategy and choose the best one when user preference is known.

$\bullet$ \textbf{Oracle Greedy Algorithm \cite{shanmugam2013femtocaching}}: SBSs Cache files greedily when user preference is known.

$\bullet$ \textbf{CUCB Algorithm \cite{Sengupta2014}}: Each SBS estimates its local file popularity by using CUCB algorithm and then the cache strategy is optimized with the estimated local popularity coordinately.

$\bullet$ \textbf{Least Frequently Used (LFU) Algorithm \cite{970573}}: Replace the file with the minimum requested times in the cache of each SBS with the requested file that is not available.

$\bullet$ \textbf{Least Recently Used (LRU) Algorithm \cite{970573}}: Replace the least recently used file in the cache of each SBS with the requested file that is not available.

In this subsection, we first show the regret of MAMAB and MAMAB-v2 algorithms in both distributed and collaborative manners. Then, we evaluate the performance of our proposed algorithms by comparing them with benchmark algorithms. Finally, we demonstrate the performances of our proposed algorithms with respect to different communication distance and cache size.


For presentation convenience, the agent-based collaborative MAMAB in SBS perspective and user perspective are denoted as ``Agent-based-1 C-MAMAB'' and ``Agent-based-2 C-MAMAB'', respectively. The edge-based collaborative MAMAB is denoted as ``Edge-based C-MAMAB''. All the plots are obtained after averaging over $30$ independent realizations.

Time evolutions for the regret of MAMAB and MAMAB-v2 algorithms in different manners are plotted in Fig. \ref{fig:M=6_v1} and Fig. \ref{fig:M=6_v3}, respectively. The regret is defined as the accumulated difference between the oracle coordinate ascent algorithm and the corresponding MAMAB algorithm. It is observed that Agent-based-1 C-MAMAB performs worst in Fig. \ref{fig:M=6_v1} due to its large action space, which causes slow convergence speed. The Edge-based C-MAMAB outperforms three other MAMAB algorithms since it considers both SBS coordination and dimensions of action space, though it has no performance guarantee. In both MAMAB and MAMAB-v2 algorithms, Edge-based ones perform best and Agent-based-2 outperforms Agent-based-1, which means the large action space degrades the performance greatly. Comparing MAMAB and MAMAB-v2 algorithms, we find that MAMAB-v2 outperforms MAMAB notably by designing a new perturbed term, which means that they can achieve a better tradeoff between exploration and exploitation\footnote{MAMAB-v2 algorithms also outperform the $\epsilon$-greedy and these results are not shown is this paper due to the page limit.}. Therefore, we only focus on MAMAB-v2 algorithms in the following parts.
\begin{figure*}[t]
\vspace{-0.3cm}
\begin{subfigure}[t]{0.5\linewidth}
\centering
\includegraphics[width=8.5cm]{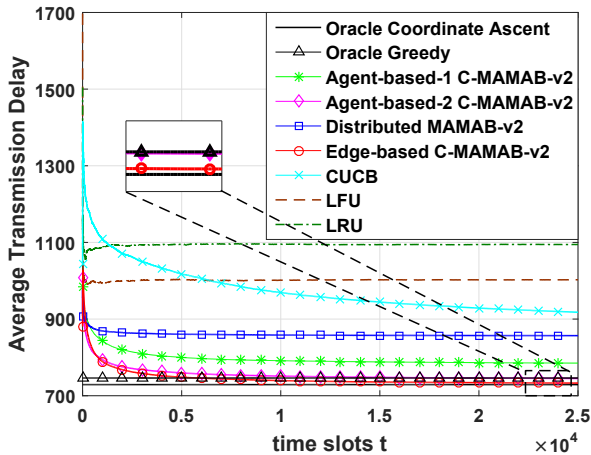}\\
\caption{{MAMAB-v2 algorithms with the same user preference.}} \label{fig:M=6_N=50_preference1}
\end{subfigure}
\hfill
\begin{subfigure}[t]{0.5\linewidth}
\centering
\includegraphics[width=8.5cm]{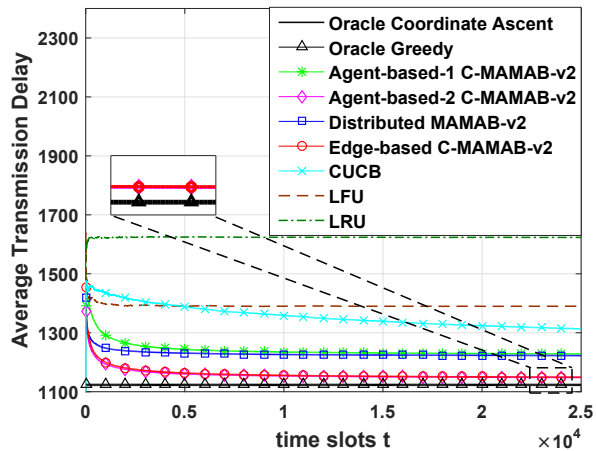}\\
\caption{MAMAB-v2 algorithms with different user preference.} \label{fig:M=8_N=100_preference1}
\end{subfigure}
\caption{{The time evolution of the average transmission delay with $M=6$, $U=50$, $l_c=50$ \text{m}, $F=100$ and $S=10$.}}
\vspace{-0.2cm}
\end{figure*}

Fig. \ref{fig:M=6_N=50_preference1} and Fig. \ref{fig:M=8_N=100_preference1} compare the performance of MAMAB-v2 with some benchmark algorithms. It is shown that the low-complexity oracle coordinate ascent algorithm outperforms oracle greedy algorithm and they both can be seen as the performance upper bound since they know the full knowledge of user preference. Besides, MAMAB-v2 algorithms outperform LRU and LFU since they concentrate on the reward that SBSs can obtain rather than users requests. Moreover, MAMAB-v2 outperforms CUCB algorithm since it distinguishes user preference and file popularity. Besides, MAMAB-v2 also has a faster learning speed than predict-then-optimize CUCB method when users have the same preference. For different MAMAB-v2 algorithms, the Edge-based C-MAMAB-v2 performs best and almost can achieve the performance of the oracle algorithms. The Agent-based-1 C-MAMAB outperforms the distributed MAMAB when users have the same preference since it considers the SBS coordination. However, when users have different preference, Agent-based-1 C-MAMAB-v2 almost performs the same as the distributed one since the correlations among SBSs are weak.

\begin{figure*}[t]
\vspace{-0.1cm}
\begin{minipage}[t]{0.5\linewidth}
\centering
\includegraphics[width=8.5cm]{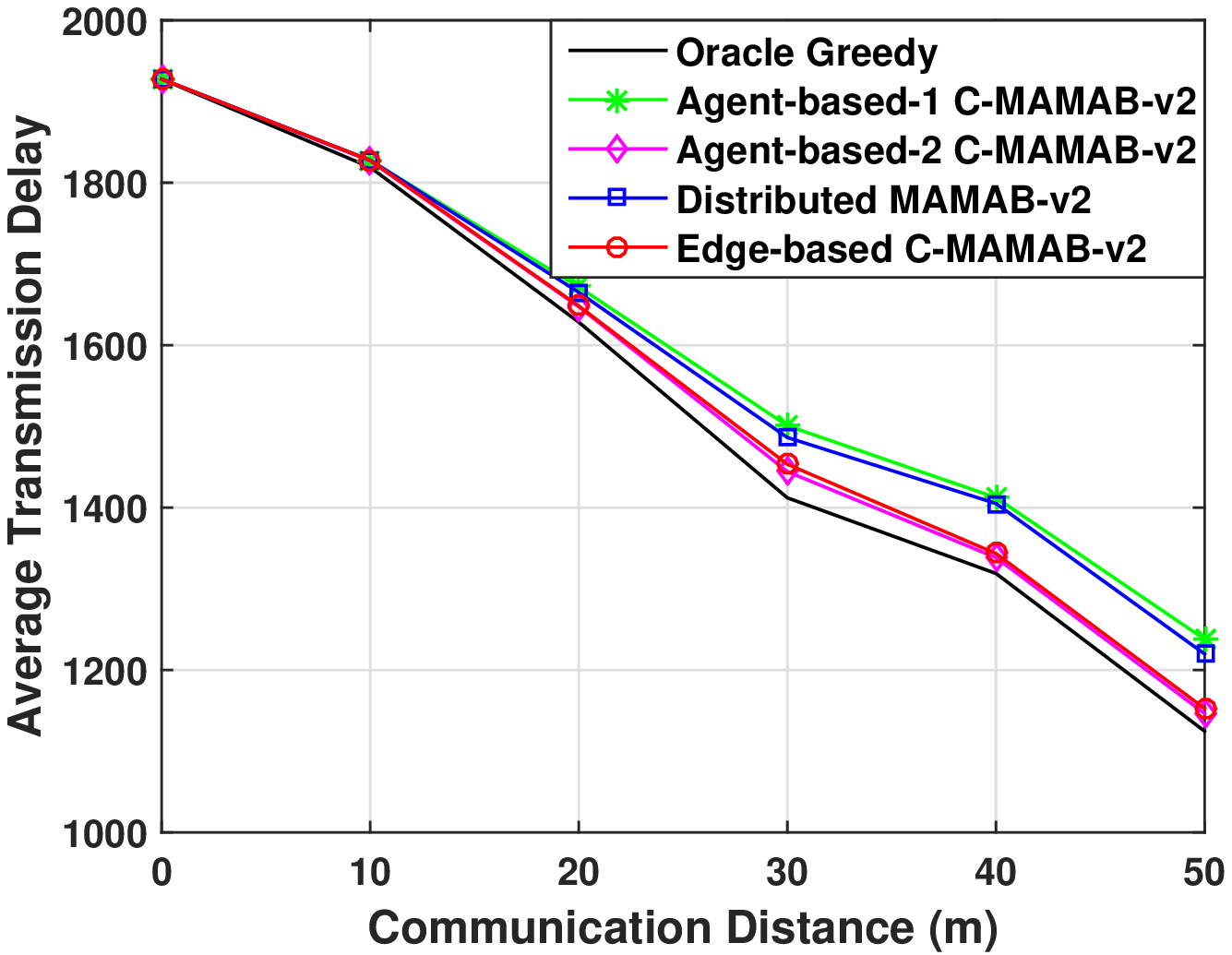}\\
\caption{{Average delay after $25000$ time slots of MAMAB-v2 algorithms vs. $l_c$ when users have different preferences with $M=6$, $U=50$ $F=100$ and $S=10$.}} \label{fig:communication_distance}
\end{minipage}
\hfill
\begin{minipage}[t]{0.5\linewidth}
\centering
\includegraphics[width=8.5cm]{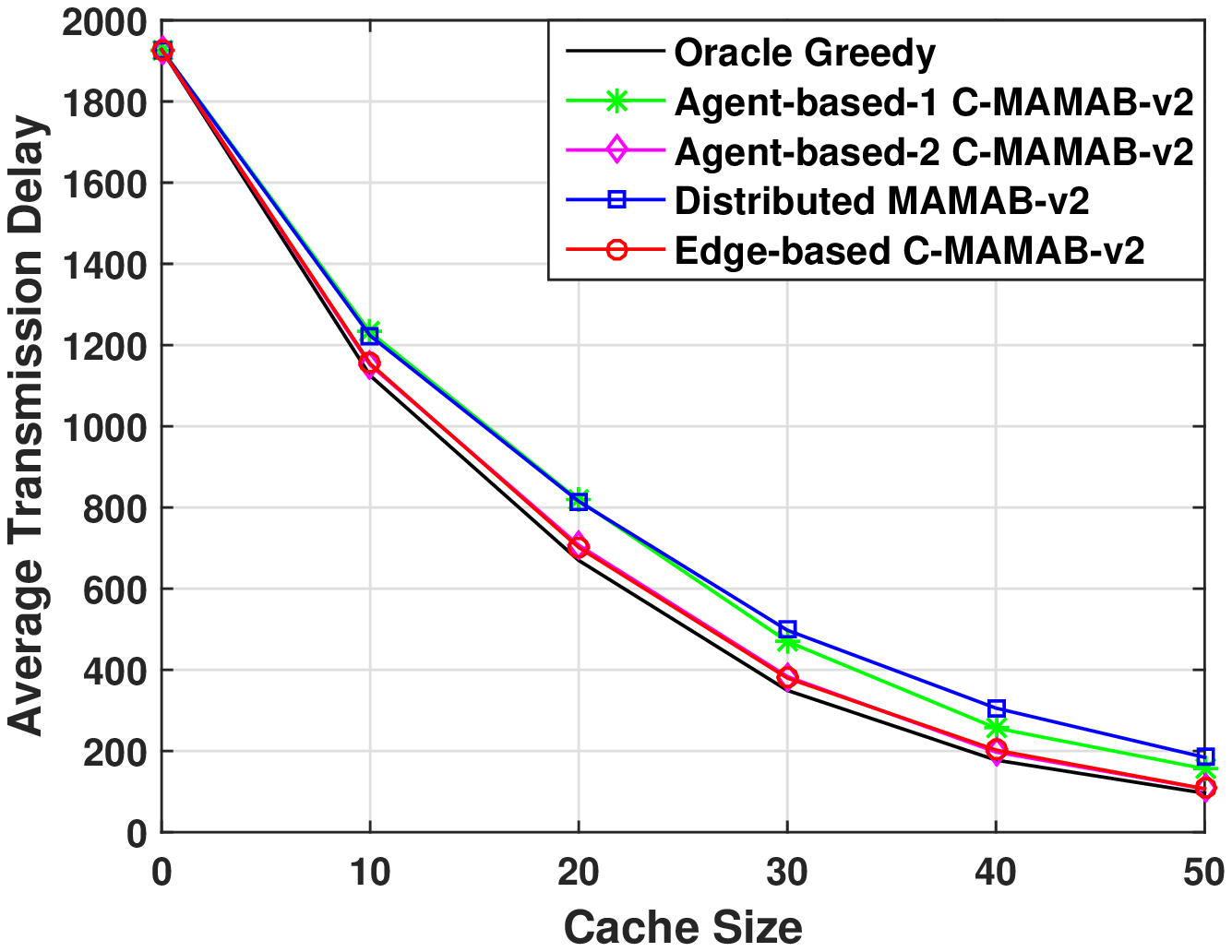}\\
\caption{{Average delay after $25000$ time slots of MAMAB-v2 algorithms vs. $S$ when users have different preferences with $M=6$, $U=50$, $F=100$ and $l_c=50$ \text{m}.}} \label{fig:cache_size}
\end{minipage}
\vspace{-0.4cm}
\end{figure*}

Fig. \ref{fig:communication_distance} illustrates the average transmission delay of different MAMAB-v2 algorithms after $25000$ time slots versus different communication distances $l_c$. By increasing $l_c$, the number of users that can be covered by multiple SBSs increases and the correlations among SBSs becomes stronger. Fig. \ref{fig:communication_distance} shows that both distributed and collaborative MAMAB-v2 algorithms almost can achieve the performance of the oracle greedy algorithm. This is because many users only have one neighbor SBS when $l_c$ is small and almost no correlations exist among SBSs. By increasing the communication distance $l_c$, it is observed that the gap between oracle greedy algorithm and distributed MAMAB-v2 increases. This is because distributed MAMAB-v2 totally ignores the correlations among SBSs while the correlations among SBSs become larger when $l_c$ grows. The gap between oracle greedy and Agent-based-1 also increases when $l_c$ grows since the performance becomes poor when {the action space} is too large. While for the Agent-based-2 and Edge-based C-MAMAB-v2, they perform close to the oracle greedy for all $l_c$ since they take both SBS coordination and action space into account.

Fig. \ref{fig:cache_size} depicts the average transmission delay of different algorithms after $25000$ time slots versus different cache size $S$. Note that the average transmission delay decreases as $S$ increases since more user requests can be satisfied by SBSs locally. However, the delay decreasing speed becomes smaller as the cache size increases, which is consistent with the Zipf distribution that a small number of popular files produce the majority of the data traffic. The Edge-based and Agent-based-2 C-MAMAB-v2 almost can achieve the performance of the oracle greedy. Comparing the performance of the Agent-based-1 C-MAMAB-v2 and distributed MAMAB-v2, it is observed that Agent-based-1 C-MAMAB-v2 performs better than distributed MAMAB-v2 when $S$ is large since it can make a better utilization of cache size by considering SBS coordination.

\subsection{Non-stationary Environment}
\begin{figure*}[t]
\vspace{-0.2cm}
\begin{minipage}[t]{0.5\linewidth}
\centering
\includegraphics[width=8.5cm]{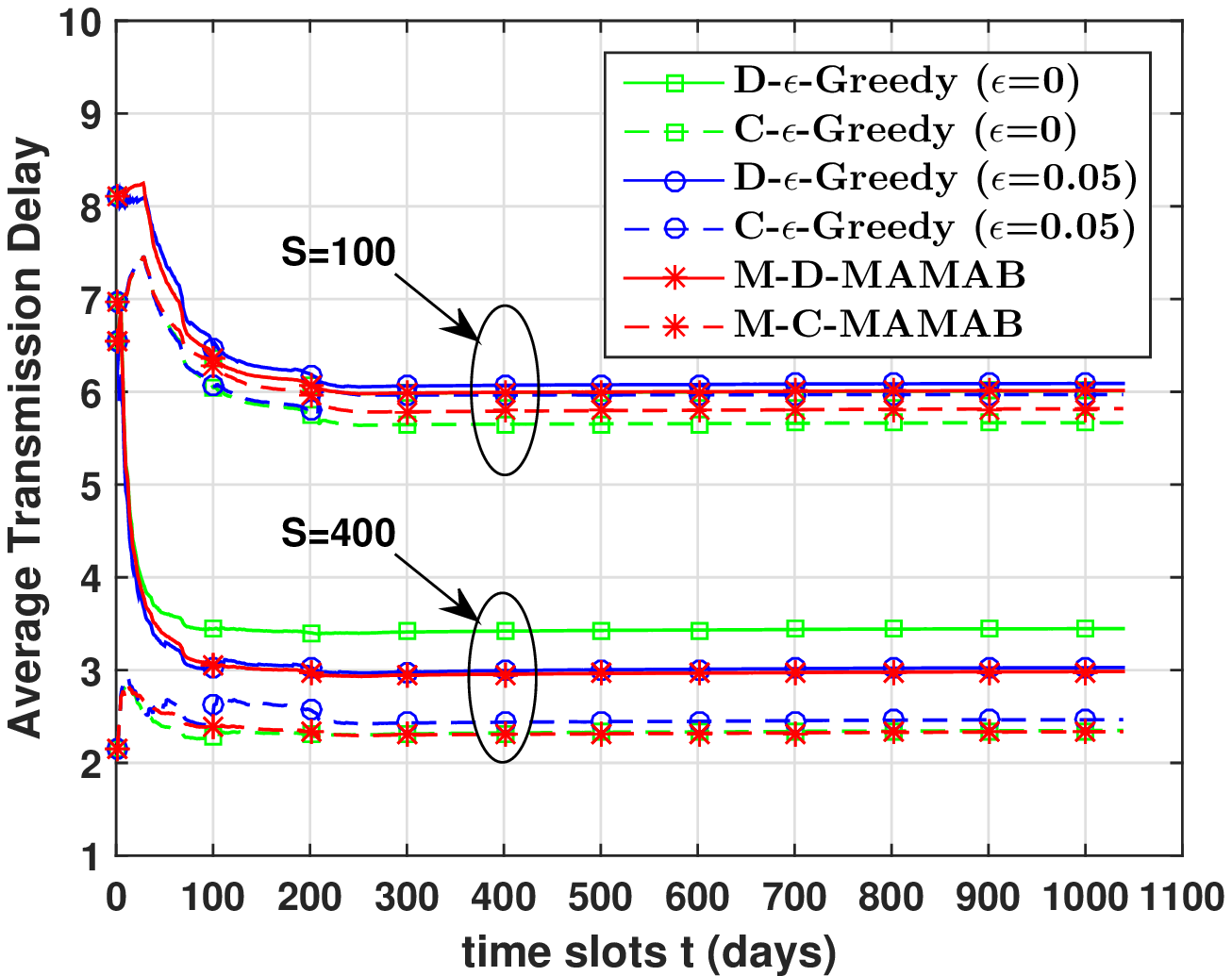}\\
\caption{{{Time evolution of algorithms for $M=5$ and $l_c=50$ $\text{m}$.}}} \label{fig:EE}
\end{minipage}
\hfill
\begin{minipage}[t]{0.5\linewidth}
\centering
\includegraphics[width=8.5cm]{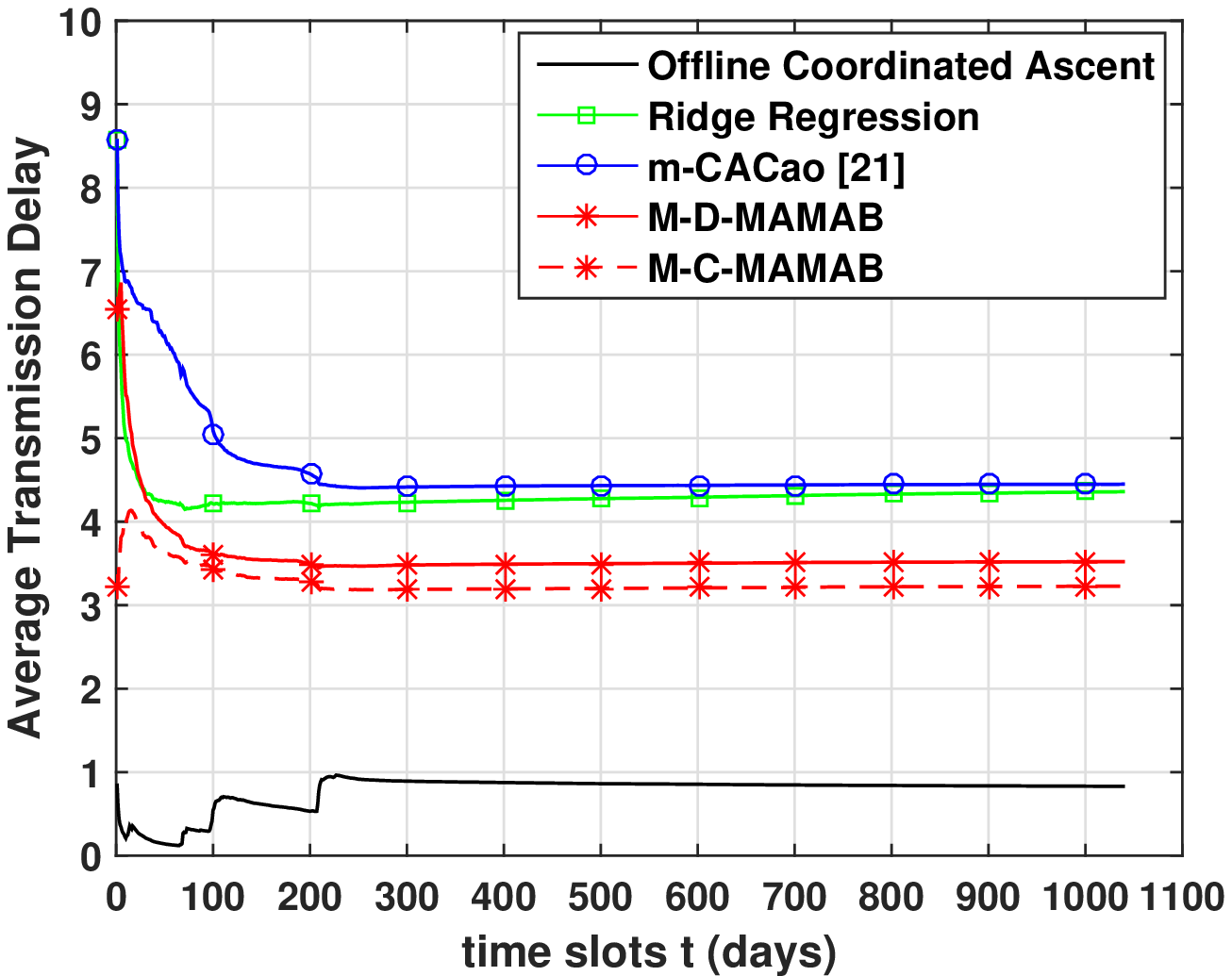}\\
\caption{{{Time evolution of algorithms for $M=5$, $l_c=40$ $\text{m}$ and $S=400$.}}} \label{fig:Modified}
\end{minipage}
\vspace{-0.3cm}
\end{figure*}

\begin{figure*}[t]
\vspace{-0.2cm}
\begin{minipage}[t]{0.5\linewidth}
\centering
\includegraphics[width=8.5cm]{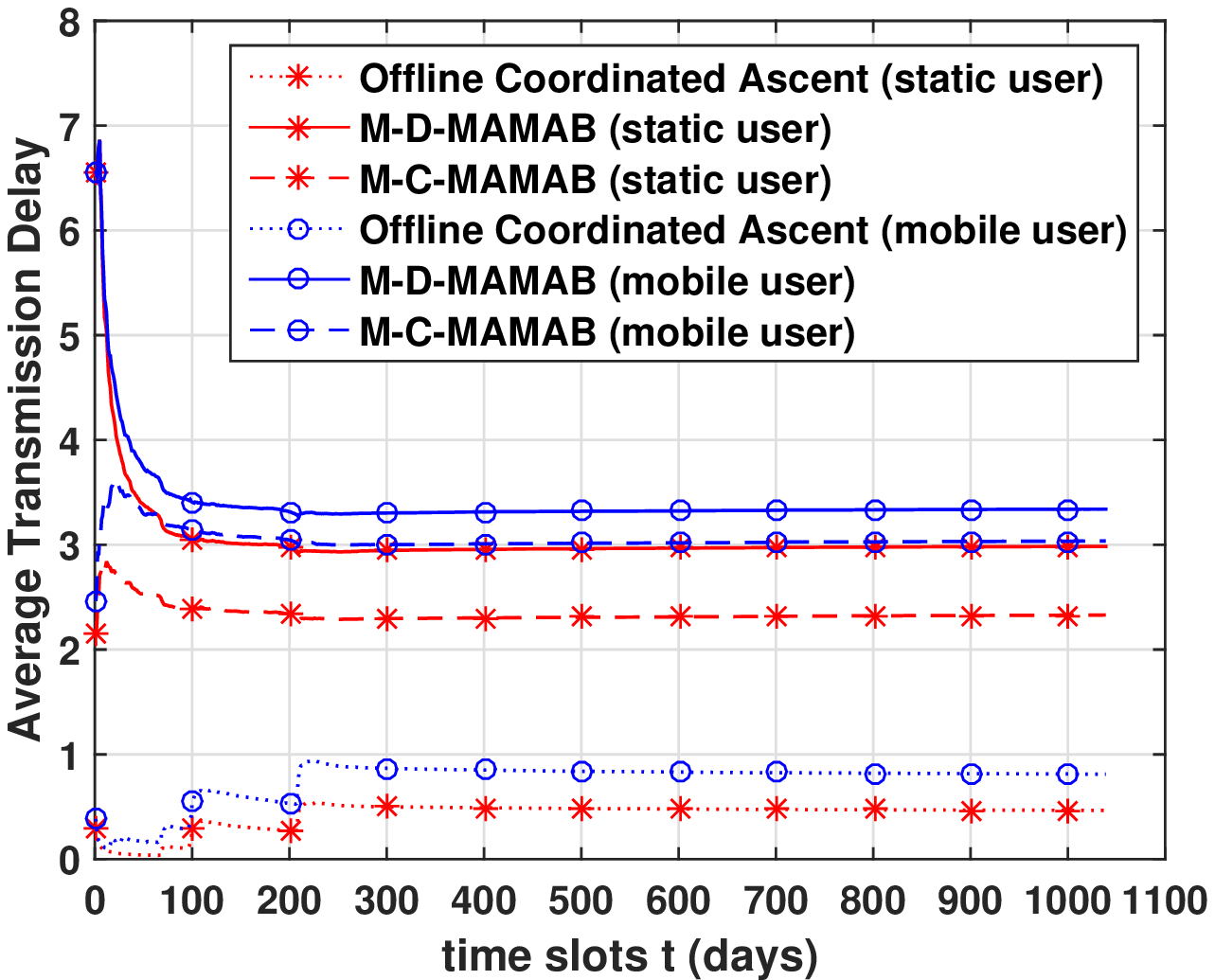}\\
\caption{{{Effects of users mobility on modified MAMAB algorithms for $M=5$, $l_c=50$ $\text{m}$ and $S=400$.}}} \label{fig:mobility}
\end{minipage}
\hfill
\begin{minipage}[t]{0.5\linewidth}
\centering
\includegraphics[width=8.5cm]{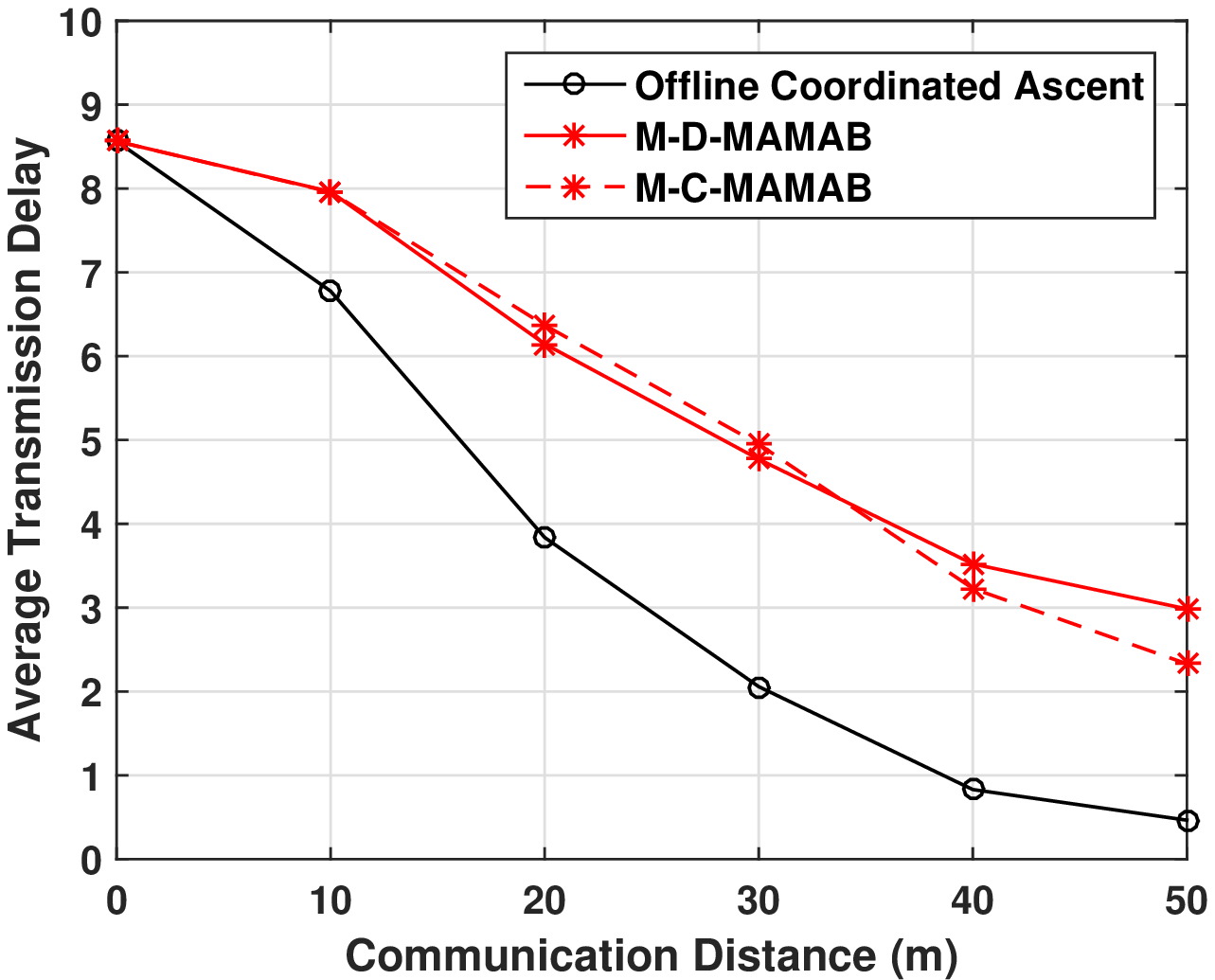}\\
\caption{{{Average delay of modified MAMAB algorithms after $1039$ days vs. communication distance for $M=5$ and $S=400$.}}} \label{fig:communication_distance_non}
\end{minipage}
\vspace{-0.3cm}
\end{figure*}
In this work, we model the non-stationary environment by using the $1M$ data set from MovieLens \cite{harper2016movielens}. This data set contains $1000209$ ratings of $|\mathcal{F}|=3952$ movies and these ratings were made by $U=6040$ users from the year $2000$ to $2003$. Due to the fact that a user only rating the movies after watching them, we assume that a rating of a user corresponds to a file request in this work. The time horizon is divided into $T_{\text{total}}=1039$ time slots (days) and SBSs update the cache once everyday. Unless otherwise stated, {Unless otherwise stated, we assume that locations of users are fixed at different time slots but the set of active users changes over time \footnote{Note that users always request files at fixed locations in practice, such as at home or office. Therefore, we only consider the dynamic active users set and ignore the mobility of users.}, and} the simulation parameters are set as follows: the number of SBSs $M=5$ {and cache size $S=400$.} We utilize the average transmission delay as the performance metric, which is defined as the ratio of the total transmission delay to the the number of requests. To evaluate the performance of our proposed algorithms, we compare them with the following benchmark algorithms:

$\bullet$ \textbf{Oracle Coordinated Ascent Algorithm}: Optimize the cache by utilizing the Algorithm \ref{alg:9} at each time slot when the user requests {and transmission delays between SBSs and users} are perfectly known in advance\footnote{The reason for choosing the coordinated ascent algorithm rather than the oracle greedy algorithm in the stationary environment is the high computation complexity of the greedy algorithm with large number of uses and files.}.

$\bullet$ {\textbf{Distributed/Edge-based Collaborative $\epsilon$-Greedy Algorithm (D/C-$\epsilon$-Greedy)}: SBSs cache files according to the $\epsilon$-greedy algorithm based on the corresponding  average reward rather than the estimated reward that contains the perturbed term.}

$\bullet$ \textbf{Ridge Regression Algorithm}: Each SBS estimates the number of requests for its neighbor users {via ridge regression by utilizing all requests information, rather than only utilizing the information of caching files as \cite{8466606},} and then caches the most popular files at each time slot.

$\bullet$ \textbf{m-CACao\cite{muller2017context}}: {SBSs} cache files according to the context-aware proactive caching algorithm.


{In this subsection, we first discuss the tradeoff between exploration and exploitation of modified MAMAB algorithms in the non-stationary environment. Then we validate the effectiveness of our proposed modified MAMAB algorithms. Finally, we discuss the effects of user mobility and communication distance $l_c$ of modified MAMAB algorithms.}
{For presentation convenience, our proposed modified MAMAB algorithm in distributed and edge-based collaborative manners are denoted as ``M-D-MAMAB'' and ``M-C-MAMAB'', respectively.}

{Fig. \ref{fig:EE} illustrates the performance of our proposed modified MAMAB algorithms and $\epsilon$-greedy algorithms with different $\epsilon$ values. It is observed that edge-based collaborative algorithms outperform the corresponding distributed algorithms. Comparing three distributed algorithms and three collaborative algorithms, we find that when the cache size is large ($S=400$), our proposed modified MAMAB algorithms with perturbed terms can achieve a better tradeoff compared to $\epsilon$-greedy algorithms. This is because large cache size can make SBSs play all actions sufficient times and hence our modified MAMAB algorithms have a fast learning speed to learn the reward of different actions accurately. However, when the cache size is small ($S=100$), it is illustrated that $\epsilon$-greedy algorithm with $\epsilon=0$ performs better than $\epsilon$-greedy algorithm with $\epsilon=0.05$ and modified MAMAB algorithms in both distributed and edge-based collaborative manners. This means that when the cache size is small, doing more exploitations is better since SBSs have a lower speed to learn the reward accurately with such a small cache size.}

{Fig. \ref{fig:Modified} illustrates the performance of our proposed modified MAMAB and some reference algorithms. It is observed that the oracle coordinate ascent performs best since it knows all information of user requests and transmission delays between users and SBSs in advance. M-C-MAMAB outperforms M-D-MAMAB since it considers the SBSs cooperation. Besides, our proposed modified MAMAB algorithms outperform the ridge regression and m-CACao since they focus on the reward rather than the request number from a set of users. Therefore, we can conclude that MAMAB is more appropriate to solve this sequential decision making problem than the ridge regression that first utilizes all requests information to predict local popularity and then caches the most popular files. Moreover, M-D-MAMAB can tackle the cooperative caching problem among SBSs better than existing algorithms.}


{In Fig. \ref{fig:mobility}, the effects of user mobility on the caching performance of modified MAMAB algorithms are evaluated. It is observed that modified MAMAB algorithms have a good performance no matter users are static or mobile. However, the gap between M-D-MAMAB and M-C-MAMAB becomes smaller when we take user mobility into account. This is because that M-D-MAMAB ignores the coordination among SBSs and has a smaller action space, and hence it has a faster learning speed to learn the reward accurately. Therefore, we can conclude that M-D-MAMAB is more robust than M-C-MAMAB when users are moving over time.}

{In Fig. \ref{fig:communication_distance_non}, the performance of modified MAMAB algorithms after $1039$ days versus $l_c$ is evaluated. By increasing $l_c$, the average transmission delay of modified MAMAB algorithms decreases since more users can be covered and served by SBSs locally as $l_c$ grows. Comparing the performance of the modified MAMAB in distributed and collaborative manners, they perform the same when $l_c$ is small ($l_c=10$ m). When $l_c$ becomes larger ($l_c=20, 30$ m), it is observed that the M-D-MAMAB even outperforms the collaborative one. This is because some joint actions of M-C-MAMAB only can serve a small number of users and have a small reward, but they still need to have enough explorations. Thus, SBSs choose these joint actions with small reward many times and lose the chance for more exploitations. But as $l_c$ keeps growing ($l_c=40,50$ m), the correlations among SBSs become stronger and the number of users that are covered by multiple SBSs becomes larger, M-C-MAMAB outperforms M-D-MAMAB since it considers the coordination among SBSs and achieves a good tradeoff between exploration and exploitation.}

\section{Conclusion}
In this work, we investigated the collaborative caching optimization problem to minimize the accumulated transmission delay without the knowledge of user preference. We model this sequential multi-agent decision making problem in a MAMAB perspective and solve it by learning cache strategy directly online in both stationary and non-stationary environment. In the stationary environment, we proposed multiple efficient MAMAB algorithms in both distributed and collaborative manners to solve the problem. In the agent-based collaborative MAMAB, we provided a strong performance guarantee by proving that its regret is bounded by $\mathcal{O}(\log{(T_{\text{total}}))}$. However, the computational complexity is high. For the distributed MAMAB, it has a low computational complexity but the correlations among SBSs are ignored. To achieve a better balance between SBS coordination and computational complexity, we also proposed a coordination graph edge-based reward assignment method in the edge-based collaborative MAMAB. In the non-stationary environment, we modified the MAMAB-based algorithms {proposed in the stationary environment by proposing a practical initialization method and designing new perturbed terms to adapt to the dynamic environment better.} Simulation results demonstrated the effectiveness of our proposed algorithms. The effects of the communication distance and cache size were also discussed.

\begin{appendices}
\section{Proof of Lemma 1}

$\textbf{A}_{\mathcal{M}\backslash \{m\}}^L=[\textbf{a}^L_1,\textbf{a}^L_2,\ldots,\textbf{a}^L_{m-1},\mathbf{0}_{F\times 1},\textbf{a}^L_{m+1},\ldots,\textbf{a}^L_M]$ is the cache strategy that SBS $m$ caches nothing and the other $M-1$ SBSs cache according to $\textbf{A}^L$. $\textbf{A}^*=[\textbf{a}^*_1,\textbf{a}^*_2,\ldots,\textbf{a}^*_M]^T$ is the optimal cache strategy. We utilize $\cup$ to combine different cache strategies together, e.g., $\textbf{a}_m^*\cup\textbf{A}^L=[\textbf{a}^L_1,\textbf{a}^L_2,\ldots,\textbf{a}^L_{m-1},\textbf{a}^L_{m}\cup\textbf{a}_m^*,\textbf{a}^L_{m+1},\ldots,\textbf{a}^L_M]$, which means that SBS $m$ caches files in both $\textbf{a}_m^*$ and $\textbf{a}_m^L$ while the other $M-1$ SBSs cache files according to $\textbf{A}^L$, and $\textbf{A}^*\cup\textbf{A}^L=[\textbf{a}^L_{1}\cup\textbf{a}_1^*,\ldots,\textbf{a}^L_{M}\cup\textbf{a}_M^*]$, which means that each SBS $m$ cache files in both $\textbf{a}_m^*$ and $\textbf{a}_m^L$. Then, we can obtain:
\begin{align}
&R_{\text{total}}(\textbf{A}^*)-R_{\text{total}}(\textbf{A}^L)\overset{(a)}{\leq} R_{\text{total}}(\textbf{A}^*\cup\textbf{A}^L)-R_{\text{total}}(\textbf{A}^L) \nonumber\\
&\overset{(b)}{\leq}\sum_{m=1}^M \left[R_{\text{total}}(\textbf{a}_m^*\cup\textbf{A}^L)-R_{\text{total}}(\textbf{A}^L)\right]   \nonumber\\
&\overset{(c)}{\leq}\sum_{m=1}^M \left[R_{\text{total}}(\textbf{a}_m^*\cup\textbf{A}_{\mathcal{M}\backslash \{m\}}^L)-R_{\text{total}}(\textbf{A}_{\mathcal{M}\backslash \{m\}}^L)\right] \nonumber\\
&\overset{(d)}{\leq}\sum_{m=1}^M \left[R_{\text{total}}(\textbf{a}_m^L\cup\textbf{A}_{\mathcal{M}\backslash \{m\}}^L)-R_{\text{total}}(\textbf{A}_{\mathcal{M}\backslash \{m\}}^L)\right]  \nonumber\\
&\overset{(e)}{\leq}\sum_{m=1}^M\left[R_{\text{total}}(\textbf{a}_m^L\cup\textbf{A}_{\mathcal{M}\backslash \{m,m+1,\ldots,M\}}^L)-R_{\text{total}}(\textbf{A}_{\mathcal{M}\backslash \{m,m+1,\ldots,M\}}^L)\right] \nonumber\\ &=R_{\text{total}}(\textbf{A}^L),
\end{align}
where step (a) follows from the fact that caching more files increases the reward. Step (b) is obtained from that for any strategies $\textbf{a}'_m$ and $\textbf{a}'_n$ and $\textbf{A}$, we have $R_{\text{total}}(\textbf{a}'_m\cup\textbf{a}_n'\cup\textbf{A})-R_{\text{total}}(\textbf{A})\leq[R_{\text{total}}(\textbf{a}_m'\cup\textbf{A})-R_{\text{total}}(\textbf{A})]+[R_{\text{total}}(\textbf{a}'_n\cup\textbf{A})-R_{\text{total}}(\textbf{A})]$
since SBS $m$ and $n$ may cover the common users. For step (c) and (e), they are obtained from $R_{\text{total}}(\textbf{a}'_m\cup\textbf{A})-R_{\text{total}}(\textbf{A})\leq R_{\text{total}}(\textbf{a}'_m\cup\textbf{A}_{\mathcal{M}\backslash\{m\}})-R_{\text{total}}(\textbf{A}_{\mathcal{M}\backslash\{m\}})$. This is because the cache strategy $\textbf{a}'_m$ and $\textbf{a}_m$ may cache the common files. Step (d) follows from that $\textbf{a}_m^L$ is the optimal cache strategy of SBS $m$ when other $M-1$ SBSs adopt according to $\textbf{A}^L_{\mathcal{M}\backslash\{m\}}$ obtained from Algorithm \ref{alg:9}. Therefore, we have $R_{\text{total}}(\textbf{A}^L)\geq\frac{1}{2}R_{\text{total}}(\textbf{A}^*)$, and the proof is completed.

\section{Proof of Lemma 2}
The proofs of two agent-based collaborative MAMAB are similar and we only focus on the agent-based collaborative MAMAB in a SBS perspective. The reward of SBS $m$ with joint action $\textbf{b}_{m,f}$ is a random variable $R_{m,f}(\textbf{b}_{m,f})\in\left[0,B_{\text{a-coll},m}(\textbf{b}_{m,f})\right]$ with expectation $\mu_{m,f}(\textbf{b}_{m,f})$. The optimal reward is denoted as $r_{\boldsymbol{\mu}} (\textbf{A}^*)=\text{opt}_{\boldsymbol{\mu}}$ respect to the the expectation vector $\boldsymbol{\mu}$.
For variable $R_{m,f}(\textbf{b}_{m,f})$, the mean value of it after this joint action $\textbf{b}_{m,f}$ appears $s$ times is $\overline{R}_{m,f}^s(\textbf{b}_{m,f})=(\sum_{i=1}^sr_{m,f}^s(\textbf{b}_{m,f}))/s$. In the $t$-th time slot, let $F_t$ be the event that the oracle fails to produce an $\alpha$-approximate answer with respect to the estimated reward $\widehat{R}_{m,f}^t(\textbf{b}_{m,f})$. We have $P[F_t]=\mathbb{E}[\mathbb{I}\{F_t\}]\leq 1-\beta$ \footnote{The coordinate ascent oracle is proven to be an $(\alpha,\beta)$ approximation with $\alpha=\frac{1}{2}$ and $\beta=1$ in Appendix A.}. The bad cache actions is denoted as $\textbf{A}_B=\{\textbf{A}|r_{\boldsymbol{\mu}} (\textbf{A})<\alpha\cdot \text{opt}_{\boldsymbol{\mu}}\}$. Then we define $\Delta_{\text{min}}=\alpha\cdot \text{opt}_{\boldsymbol{\mu}}-\underset{\textbf{A}\in\textbf{A}_B}{\text{max}}\{r_{\boldsymbol{\mu}}(\textbf{A})\}$ and $\Delta_{\text{max}}=\alpha\cdot \text{opt}_{\boldsymbol{\mu}}-\underset{\textbf{A}\in\textbf{A}_B}{\text{min}}\{r_{\boldsymbol{\mu}}(\textbf{A})\}$.

We maintain counter $N_{m,f}(\textbf{b}_{m,f})$ for each joint action $\textbf{b}_{m,f}$ and we have $N_{m,f}^{0}(\textbf{b}_{m,f})=1$ since each joint action $\textbf{b}_{m,f}$ appear once in the initial phase. At each time slot $t$, if we choose a bad action $\textbf{A}_B$, we increase one minimum counter $N_{m,f}(\textbf{b}_{m,f})$ by one at this bad time slot, i.e., $N^{t}_{m,f}(\textbf{b}_{m,f})=N^{t-1}_{m,f}(\textbf{b}_{m,f})+1$. Therefore, we have $N^{t}_{m,f}(\textbf{b}_{m,f})\leq T^{t}_{m,f}(\textbf{b}_{m,f})$.

Define $Q_t=\underset{\textbf{b}_{m,f}}{\text{max}}B_{\text{a-coll},m}^2(\textbf{b}_{m,f})\frac{6 \log t}{(f^{-1}(\Delta_{\text{min}}))^2}$, where $f(\cdot)$ is the bounded smoothness function\cite{chen2013combinatorial}. For a bad time slot $t$ with action $\textbf{A}^t\in\textbf{A}_B$, we have
\begin{align}
&\sum_{m=1}^M\sum_{f=1}^F\sum_{\textbf{b}_{m,f}} N^{T_{\text{total}}}_{m,f}(\textbf{b}_{m,f})-\sum_{m=1}^M\sum_{f=1}^F\sum_{\textbf{b}_{m,f}}Q_{T_{\text{total}}}  \nonumber\\
&=\sum_{t=1}^{T_{\text{total}}} \mathbb{I}\{\textbf{A}^t\in\textbf{A}_B\}+\sum_{m=1}^M2^{\Gamma(m)}F-\sum_{m=1}^M\sum_{f=1}^F\sum_{\textbf{b}_{m,f}}Q_{T_{\text{total}}} \nonumber\\
&\leq \sum_{t=1}^{T_{\text{total}}}\sum_{m=1}^M\sum_{f=1}^F\sum_{\textbf{b}_{m,f}}\mathbb{I}\{\textbf{A}^t\in\textbf{A}_B,N^t_{m,f}(\textbf{b}_{m,f})>N_{m,f}^{t-1}(\textbf{b}_{m,f}),\nonumber\\
&~~~~~~~~~~~~~~~~~~~~~~~~~~~N_{m,f}^{t-1}(\textbf{b}_{m,f})>Q_{T_{\text{total}}}\}+M2^{M-1}F \nonumber\\
&\leq \sum_{t=1}^{T_{\text{total}}}\sum_{m=1}^M\sum_{f=1}^F\sum_{\textbf{b}_{m,f}}\mathbb{I}\{\textbf{A}^t\in\textbf{A}_B,N^t_{m,f}(\textbf{b}_{m,f})>N_{m,f}^{t-1}(\textbf{b}_{m,f}), \nonumber\\
&~~~~~~~~~~~~~~~~~~~~~~~~~~~N_{m,f}^{t-1}(\textbf{b}_{m,f})>Q_t\}+M2^{M-1}F  \nonumber\\
&=\sum_{t=1}^{T_{\text{total}}}  \mathbb{I}\{\textbf{A}^t\in\textbf{A}_B, \forall ~\textbf{b}^t_{m,f},N_{m,f}^{t-1}(\textbf{b}^t_{m,f})>Q_t\}+M2^{M-1}F \nonumber\\
&\leq \sum_{t=1}^{T_{\text{total}}} \left(\mathbb{I}\{F_t\}+\mathbb{I}\{\neg F_t, \textbf{A}^t\in\textbf{A}_B, \forall~ \textbf{b}^t_{m,f},N_{m,f}^{t-1}(\textbf{b}^t_{m,f})>Q_t\}\right)\nonumber\\
&~~~~~~~~~~~~~~~~~~~~~~~~~~~+M2^{M-1}F \nonumber\\
&\leq \sum_{t=1}^{T_{\text{total}}} \left(\mathbb{I}\{F_t\}+\mathbb{I}\{\neg F_t, \textbf{A}^t\in\textbf{A}_B, \forall~ \textbf{b}^t_{m,f},T_{m,f}^{t-1}(\textbf{b}^t_{m,f})>Q_t\}\right)\nonumber\\
&~~~~~~~~~~~~~~~~~~~~~~~~~~~+M2^{M-1}F.
\end{align}

Note that we have
\begin{align}
&P\bigg[|\overline{R}_{m,f}^{T_{m,f}^{t-1}(\textbf{b}_{m,f})}(\textbf{b}_{m,f})-\mu_{m,f}(\textbf{b}_{m,f})| \nonumber\\
&~~~~~~~~~~~~~~~~\geq B_{\text{a-coll},m}(\textbf{b}_{m,f})\sqrt{{3\log t}/{2T_{m,f}^{t-1}(\textbf{b}_{m,f})}}\bigg] \nonumber\\
&=\sum_{s=1}^{t}P\left[T_{m,f}^{t-1}(\textbf{b}_{m,f})=s\right]  \nonumber\\
&P\left[|\overline{R}_{m,f}^s(\textbf{b}_{m,f})-\mu_{m,f}(\textbf{b}_{m,f})|\geq B_{\text{a-coll},m}(\textbf{b}_{m,f})\sqrt{3\log t/2s}\right]\nonumber\\
&\leq \sum_{s=1}^{t} P\big[|\overline{R}_{m,f}^s(\textbf{b}_{m,f})-\mu_{m,f}(\textbf{b}_{m,f})|  \nonumber\\
&~~~~~~~~~~~~~~~~\geq B_{\text{a-coll},m}(\textbf{b}_{m,f})\sqrt{{3\log t}/{2s}}\big] \nonumber\\
&\leq 2t^{-2},
\end{align}
where the last inequality comes from the Chernoff-Hoeffding Inequality.

Define a random variable $\Lambda_{m,f}^t(\textbf{b}_{m,f})=B_{\text{a-coll},m}(\textbf{b}_{m,f})\sqrt{{3\log t}/{2T_{m,f}^{t-1}(\textbf{b}_{m,f})}}$ and event

$\!\!\!\!\!\!\!E_t=\{\forall ~\textbf{b}_{m,f}, |\overline{R}_{m,f}^{T_{m,f}^{t-1}(\textbf{b}_{m,f})}(\textbf{b}_{m,f})-\mu_{m,f}(\textbf{b}_{m,f})|\leq \Lambda_{m,f}^t(\textbf{b}_{m,f})\}$. We have $P(\neg E_t)\leq M2^{M}Ft^{-2}$. From (\ref{eqn:est_agent1}), we have $\widehat{R}_{m,f}^{t}(\textbf{b}_{m,f})-\overline{R}_{m,f}^{T_{m,f}^{t-1}(\textbf{b}_{m,f})}(\textbf{b}_{m,f})=\Lambda_{m,f}^t(\textbf{b}_{m,f})$. Thus, $|\overline{R}^{T_{m,f}^{t-1}(\textbf{b}_{m,f})}(\textbf{b}_{m,f})-\mu_{m,f}(\textbf{b}_{m,f})|\leq \Lambda^t(\textbf{b}_{m,f})$ implies that $\widehat{R}_{m,f}^{t}(\textbf{b}_{m,f})\geq \mu_{m,f}(\textbf{b}_{m,f})$. In other words, we have $E_t\Longrightarrow  \widehat{\textbf{r}}^{t}\geq \boldsymbol{\mu}$, where $\widehat{\textbf{r}}^{t}$ is the vector of $\widehat{R}_{m,f}^{t}(\textbf{b}_{m,f})$ for all $\textbf{b}_{m,f}$.

Let $\Lambda=\underset{\textbf{b}_{m,f}}{\text{max}}B_{\text{a-coll},m}(\textbf{b}_{m,f})\sqrt{\frac{3\log t}{2Q_t}}$ and define $\Lambda^t=\text{max}\{\Lambda_{m,f}^t(\textbf{b}^t_{m,f})|\forall~\textbf{b}^t_{m,f}\}$. Then we have
\begin{align}
&E_t\Longrightarrow |\widehat{R}_{m,f}^{t}(\textbf{b}^t_{m,f})-\mu_{m,f}(\textbf{b}^t_{m,f})|\leq2 \Lambda^t, \forall~ \textbf{b}^t_{m,f},   \\
&\{\textbf{A}^t\in \textbf{A}_B, \forall~ \textbf{b}^t_{m,f}, T_{m,f}^{t-1}(\textbf{b}^t_{m,f})>Q_t\} \Longrightarrow \Lambda>\Lambda^t.
\end{align}

Therefore, if $\{E_t,\neg F_t,\textbf{A}^t\in \textbf{A}_B, \forall~\textbf{b}^t_{m,f}, T_{m,f}^{t-1}(\textbf{b}^t_{m,f})>Q_t\}$ holds at time slot $t$, we have
\begin{align}
&~~r_{\boldsymbol{\mu}}(\textbf{A}^t)+f(2\Lambda)>r_{\boldsymbol{\mu}}(\textbf{A}^t)+f(2\Lambda^t)\geq r_{\widehat{\textbf{r}}^{t}}(\textbf{A}^t)\geq \alpha \cdot \text{opt}_{\widehat{\textbf{r}}^{t}} \nonumber\\
&\geq \alpha \cdot r_{\widehat{\textbf{r}}^{t}}(\textbf{A}^*)  \geq \alpha \cdot r_{\boldsymbol{\mu}}(\textbf{A}^*)= \alpha \cdot \text{opt}_{\boldsymbol{\mu}}.
\end{align}

Since we have $f(2\Lambda)=\Delta_{\text{min}}$, $r_{\boldsymbol{\mu}}(\textbf{A}^t)+f(2\Lambda)>\alpha \cdot \text{opt}_{\boldsymbol{\mu}}$ contradicts the definition of $\Delta_{\text{min}}$. Therefore, we have $P[\{E_t,\neg F_t,\textbf{A}^t\in \textbf{A}_B, \forall~\textbf{b}^t_{m,f}, T_{m,f}^{t-1}(\textbf{b}^t_{m,f})>Q_t\}]=0$. Then we can conclude that $P[\{\neg F_t,\textbf{A}^t\in \textbf{A}_B, \forall~\textbf{b}^t_{m,f}, T_{m,f}^{t-1}(\textbf{b}^t_{m,f})>Q_t\}\}] \leq P[\neg E_t]\leq M2^MFt^{-2}$.
Then we have
\begin{align}
&\mathbb{E}[\sum_{m=1}^M\sum_{f=1}^F\sum_{\textbf{b}_{m,f}} N^{T_{\text{total}}}_{m,f}(\textbf{b}_{m,f})] \leq \sum_{m=1}^M\sum_{f=1}^F\sum_{\textbf{b}_{m,f}}Q_{T_{\text{total}}}  \nonumber\\
&~~~~~~~~~~~~~~+(1-\beta)T_{\text{total}}+\sum_{t=1}^{T_{\text{total}}}M2^MFt^{-2}+M2^{M-1}F \nonumber\\
&\leq \underset{\textbf{b}_{m,f}}{\text{max}}B_{\text{a-coll},m}^2(\textbf{b}_{m,f})\frac{6M2^{M-1}F \log T_{\text{total}}}{(f^{-1}(\Delta_{\text{min}}))^2} \nonumber\\
&~~~~~~~~~~~~~~+M2^{M-1}F(\frac{\pi^2}{3}+1)+(1-\beta)T_{\text{total}}.
\end{align}

Therefore, the regret is given by
\begin{align}
&Z_{\text{regret}}=\sum_{t=1}^{T_{\text{total}}}\mathbb{E}[R_{\text{total}}(A^*)-R_{\text{total}}(A^t)] \nonumber\\
&\leq T_{\text{total}}\alpha\beta \text{opt}_{\boldsymbol{\mu}}-\bigg[T_{\text{total}}\alpha \text{opt}_{\boldsymbol{\mu}} \nonumber\\
&~~~~~~~-\mathbb{E}[\sum_{m=1}^M\sum_{f=1}^F\sum_{\textbf{b}_{m,f}} N^{T_{\text{total}}}_{m,f}(\textbf{b}_{m,f})-M2^{M-1}F]\cdot \Delta_{\text{max}}\bigg] \nonumber\\
&\leq \bigg[\underset{\textbf{b}_{m,f}}{\text{max}}B_{\text{a-coll},m}^2(\textbf{b}_{m,f})\frac{6M2^{M-1}F \log T_{\text{total}}}{(f^{-1}(\Delta_{\text{min}}))^2}+\frac{M2^MF\pi^2}{6} \nonumber\\
&~~~~~~~+(1-\beta)T_{\text{total}}\bigg]\Delta_{\text{max}}-(1-\beta)T_{\text{total}}\alpha \text{opt}_{\boldsymbol{\mu}}\nonumber\\
&\leq \!\left[\!\underset{\textbf{b}_{m,f}}{\text{max}}B_{\text{a-coll},m}^2(\textbf{b}_{m,f})\frac{6M2^{M-1}F \log T_{\text{total}}}{(f^{-1}(\Delta_{\text{min}}))^2}\!+\!\frac{M2^MF\pi^2}{6}\!\right]\!\!\Delta_{\text{max}}.
\end{align}
Therefore, the proof is completed. For the agent-based collaborative MAMAB in a user perspective, the proof is similar and hence is omitted here.
\end{appendices}
\bibliographystyle{IEEEtran}
\bibliography{IEEEabrv,MMAB}

\end{document}